\documentclass[a4paper]{article}
\usepackage{pdfpages}
 \usepackage{fullpage}
\usepackage{graphicx,amssymb,amsmath,amsthm}
\usepackage{pict2e}
\usepackage{xspace}
\usepackage{hyperref}

\newtheorem{theorem}{Theorem} 
\newtheorem{lemma}{Lemma}
\newtheorem{claim}{Claim}
\newtheorem{corollary}{Corollary}
\newtheorem{conjecture}{Conjecture}
\newtheorem{proposition}{Proposition}

\newtheorem{definition}{Definition}

\newcommand{\etal}{{\em et al.}}

\newcommand{\container}{covering\xspace}
\newcommand{\containers}{coverings\xspace}
\newcommand{\containment}{covering\xspace}

\newcommand{\brimful}{brimful\xspace}

\renewcommand{\phi}{\varphi}
\newcommand{\chull}{\ensuremath{\Phi}}
\newcommand{\segset}{\ensuremath{S_{\textsf{seg}}}}
\newcommand{\symset}{\ensuremath{S_{\textsf{sym}}}}
\newcommand{\wormset}{\ensuremath{S_{\textsf{worm}}}}
\newcommand{\ccset}{\ensuremath{S_{\textsf{c}}}}
\newcommand{\triset}{\ensuremath{S_{\textsf{t}}}}

\newcommand{\gt}{\ensuremath{\Gamma_{\textsf{t}}}}
\newcommand{\gcc}{\ensuremath{\Gamma_3}}
\usepackage{color}

\newcommand{\peri}[1]{\ensuremath{\ell(#1)}}

\DeclareRobustCommand{\irtriangle}{%
  \begingroup
  \setlength{\unitlength}{1ex}%
  \begin{picture}(1,1)
  \polyline(1.5,0)(0,0)(1.5,1.5)(1.5,0)(.5,0)
  \end{picture}%
  \endgroup
}

\title{Universal convex \container problems under translation and discrete rotations\thanks{Work by M. K. Jung and H.-K. Ahn were supported by the Institute of Information \& communications Technology Planning \& Evaluation(IITP) grant funded by the Korea government(MSIT) (No. 2017-0-00905, Software Star Lab (Optimal Data Structure and Algorithmic Applications in Dynamic Geometric Environment)) and (No. 2019-0-01906, Artificial Intelligence Graduate School Program(POSTECH)). S.D.Yoon was supported by ``Cooperative Research Program for Agriculture Science \& Technology Development (Project No. PJ015269032022)'' Rural Development Administration, Republic of Korea. 
Work by T. Tokuyama was partially supported by MEXT JSPS Kakenhi 20H04143.}
}

\author{Mook Kwon Jung\thanks{Department of Computer Science and Engineering, Pohang University of Science and Technology, Pohang, Korea. \texttt{jmg1032@postech.ac.kr}}
\and Sang Duk Yoon\thanks{Department of Service and Design Engineering, SungShin Women's University, Seoul, Korea. \texttt{sangduk.yoon@sungshin.ac.kr}}
\and Hee-Kap Ahn\thanks{Graduate School of Artificial Intelligence, Department of Computer Science and Engineering, Pohang University of Science and Technology, Pohang, Korea. \texttt{heekap@postech.ac.kr}}
\and Takeshi Tokuyama\thanks{Department of Computer Science, School of Engineering, Kwansei Gakuin University, Sanda, Japan. \texttt{tokuyama@kwansei.ac.jp}}}

\begin{document}
\date{}
\maketitle

\begin{abstract}
We consider the smallest-area universal \container of planar objects of perimeter $2$ (or equivalently closed curves of length $2$) allowing translation and discrete rotations.
In particular, we show that the solution is an equilateral triangle of height $1$ when translation and discrete rotation of $\pi$ are allowed. 
Our proof is purely geometric and elementary.
We also give convex \containers of closed curves of length $2$ under translation 
and discrete rotations of multiples of $\pi/2$ and $2\pi/3$. 
We show a minimality of the \container for discrete rotation of multiples of $\pi/2$, 
which is an equilateral triangle of height smaller than $1$, and
conjecture that the \container  is the smallest-area convex \container. 
Finally, we give the smallest-area convex \containers of all unit segments under translation and discrete rotations $2\pi/k$ for all integers $k\ge 3$.
\end{abstract}

\section{Introduction}
Given a (possibly infinite) set $S$ of planar objects and a group $G$ 
of geometric transformations, a $G$-\container $K$ of $S$ is  a region such that 
every object in $S$ can be contained in $K$ by transforming the object with a suitable transformation $g \in G$.
Equivalently, every object of $S$ is contained in $g^{-1}K$
for a suitable transformation $g \in G$.
That is,
\[ \forall \gamma  \in  S,\; \exists g  \in G \:\text{such that}\:  g \gamma  \subseteq K.\]
We denote the group of planar translation by $T$ and that of planar translation and rotation by $TR$.  Mathematically, $TR = T \rtimes R$ is the semidirect product of $T$ and the rotation group $R  = SO(2,\mathbb{R})$.
We often call \containers for $G$-\containers if  $G$ is known from the context.

The problem of finding a smallest-area \container  is  a classical problem in mathematics,
and such a covering is often called a {\em universal \container}.
In the literature, the cases where $G = T$ or $G=TR$ have been widely studied.

The universal \container problem has attracted many mathematicians. 
Henri Lebesgue (in his letter to J. P\'{a}l in 1914) proposed a problem to find 
the smallest-area convex $TR$-\container of all objects of unit diameter (see \cite{BMP2005,BBG2015,G2018} for its history).
Soichi Kakeya considered in 1917 the $T$-\container  of the set $\segset$ of all unit line segments (called needles)~\cite{Kakeya}.
Precisely, his formulation is to find the smallest-area region in which 
a unit-length needle can be turned round, but it is equivalent to the \containment problem if the \container is convex~\cite{Bae2018}.
Originally, Kakeya considered the convex \container, and Fujiwara conjectured that the equilateral triangle of height 1 is the solution.  The conjecture was affirmatively solved by P\'{a}l in 1920 \cite{Pal}.
For the nonconvex \container, Besicovitch~\cite{B1928} gave a construction such that the area can be arbitrarily small.

Generalizing P\'{a}l's result, for any set of $n$ segments, there is a triangle to be the smallest-area convex $T$-\container of the set, and the triangle can be computed 
efficiently in $O( n \log n )$ time~\cite{Ahn2014}.
It is further conjectured that the smallest-area convex $TR$-\container of a family of triangles is a triangle,
which is shown to be true for some families~\cite{Park2021}. 

The problem of  finding  the smallest-area \container and convex \container  of the set of all curves of unit length and $G =TR$ was given by Leo Moser as an open problem in 1966 \cite{Moser66}.   
The problem is still unsolved, and the best lower bound of the smallest area of the convex \container is 
slightly larger than $0.23$, while the best upper bound was about  $0.27$ for long time~\cite{CFG1991,BMP2005}.
Wetzel informally conjectured (formally published in~\cite{Wetzel1}) in 1970 that the $30^\circ$ circular fan of unit radius, which has an area $\pi/12 \approx 0.2618$, is a convex $TR$-\container of all unit-length curves, 
and it is recently proved by Paraksa and Wichiramala~\cite{PW2021}.
However,  when only translations are allowed, 
the equilateral triangle of height $1$ is the smallest-area convex \container (Corollary~\ref{cor:t.container.worm}), which  is the same as the case of 
considering the unit line segments. 
Corollary~\ref{cor:t.container.worm} is folklore and the authors cannot find its concrete proof in the literature.
The fact can be confirmed analytically, since it suffices to consider polyline worms with two legs. Also, we can directly prove it by applying a reflecting argument similar to the proof of Theorem~\ref{thm:G2.main}, which is given in Appendix.

There are many variants of Moser's problem, and they are called {\it  Moser's worm problem}.  
The history of progress on the topic can be found in an  article~\cite{Moser91} by 
William Moser (Leo's younger brother), in Chapter D18 in~\cite{CFG1991},
and in Chapter 11.4 in~\cite{BMP2005}.
It is interesting to find a new case of Moser's worm problem with a clean 
mathematical solution. 

From now on, we focus on the convex $G$-\container of  the set $\ccset$ of all closed curves $\gamma$  of  length $2$. Here, we follow the tradition of previous works on this problem to consider length $2$ instead of $1$, since a unit line segment 
can be considered as a degenerate convex closed curve of length $2$. 
Since the boundary curve of the convex hull $C(\gamma)$ of any closed curve $\gamma$ 
is not longer than $\gamma$, it suffices to consider only convex curves.

This problem is known to be an interesting but hard variant of Moser's worm problem, and remains unsolved for $T$ and $TR$ despite of substantial efforts in the literature~\cite{FW2011,Wetzel1,Wetzel2,CFG1991,BMP2005}.
As far as the authors know, the smallest-area
convex $TR$-\container known so far is a hexagon obtained by clipping two corners of a rectangle
given by Wichiramala, and its area is slightly less than $0.441$~\cite{W2018}.  It is also shown that the smallest area is at least $0.39$~\cite{GS2020a}, which has been recently improved to  $0.4$~\cite{GS2020b} with help of computer programs. 
The smallest area of the convex $T$-\container is known to be between $0.620$ and $0.657$~\cite{BMP2005}.

There are some works on restricted shapes of \container. 
Especially, if we consider triangular \containers, Wetzel~\cite{Wetzel2,Wetzel3} gave a
complete description, and it is shown that an acute triangle with side lengths $a$, $b$, $c$ and area $X$ becomes a $T$-\container (resp. $TR$-\container) of $\ccset$ 
if and only if $2 \le \frac{8X^2}{abc}$ (resp. $2 \le \frac{2\pi X}{a+b+c}$).  As a consequence, the equilateral triangle of side length $4/3$ (resp. $\frac{2\sqrt{3}}{\pi}$) 
is the smallest triangular $T$-\container (resp. $TR$-\container) of $\ccset$.  Unfortunately, their areas are larger than those of the known smallest-area convex \containers.

If $H$ is a subgroup of $G$, an $H$-\container is a $G$-\container.  
Since $T \subset TR$, it is quite reasonable to consider groups $G$ 
lying between them, that is $T \subset G \subset TR$. 
The group $R = SO(2, \mathbb{R})$ is an abelian group, and its finite subgroups 
are $Z_k = \{ e^{2i\pi \sqrt{-1} /k} \mid 0 \le i \le k-1 \}$ for  $k=1,2,\ldots,$ where $e^{\theta \sqrt{-1} }$ means the rotation of angle $\theta$.

In this paper, we consider the \containers  under the action of the group $G _k= T \rtimes Z_k$.
We show that  the smallest-area  convex $G_2$-\container of $\ccset$ is the equilateral triangle of height $1$, whose area is $\frac{\sqrt{3}}{3}\approx 0.577.$ 
A nice feature is that the proof is purely geometric and elementary, assuming Pal's result mentioned above. 

Then, we show that the equilateral triangle with height $\beta = \cos (\pi/12) \approx 0.966$ is a  $G_4$-\container of $\ccset$.
Its area is $\frac{2\sqrt{3} +3 }{12} \approx  0.538675$, and we 
conjecture that it is the smallest-area  convex $G_4$-\container.
It is a pleasant surprise that the equilateral triangular \container becomes the optimal convex \container 
if we consider rotation by $\pi$ ($G_2$-\container),  
and it also seems to be
true if we consider rotations by $\pi/2$ ($G_4$-\container).

However, the minimum area convex $G_3$-\container is no longer an equilateral triangle. 
We give a convex $G_3$-\container of $\ccset$ 
which has area not larger than $0.568$.
Our $G_3$-\container has areas strictly smaller than 
the area of the $G_2$-\container $\triangle_1$, and a bit larger than the area of 
the $G_4$-\container $\triangle_\beta$. Unlike the $G_2$- and $G_4$-\containers, 
the $G_3$-\container is not \emph{regular} under rotation.
We show that any convex $G_3$-\container of $\ccset$ which is regular 
under rotation by $\pi/2$ or $2\pi/3$ has area strictly larger than 
the area of our $G_3$-\container. 
For triangles of perimeter $2$, we give a smaller convex $G_3$-\container
of them that has area $0.563$.

We also determine the set of all
smallest-area convex $G_k$-\containers of $\segset$, which are all triangles.
\medskip

We use $\peri{C}$ to denote the perimeter of a compact set $C$,
and $\peri{\gamma}$ to denote the length of a curve $\gamma$.
The \emph{slope} of a line is the angle swept from the $x$-axis in a counterclockwise direction
to the line, and it is thus in $[0,\pi)$. The slope of a segment is the slope of the line that
extends it. 
For two points $p$ and $q$, we use $pq$ to denote
the line segment connecting $p$ and $q$, and by $|pq|$ the length of $pq$.
We use $\ell_{pq}$ to denote the line passing through $p$ and $q$.

\section{Covering under rotation by 180 degrees}
In this section, we show that the smallest-area convex $G_2$-\container of the set $\ccset$
of all closed curves of length $2$  is the equilateral triangle of  height $1$, denoted by $\triangle_1$, 
whose area is $\sqrt{3}/3$. 
\subsection{The smallest-area \container and related results}
First, we recall a famous result mentioned in the introduction.
\begin{theorem} [Pal's theorem for the convex Kakeya problem]
\label{thm:pal.kakeya}
The equilateral triangle $\triangle_1$ is the smallest-area convex $T$-\container of the set 
of all unit line segments.
\end{theorem}

\begin{corollary}\label{cor:G2.area.bound}
The area of a  convex  $G_2$-\container of $\ccset$ is at least $\sqrt{3}/3$.
\end{corollary}
\begin{proof}
Observe that all unit line segments are in $\ccset$, and 
line segments are stable under the action of rotation by $\pi$.  Thus, any convex
$G_2$-\container of $\ccset$  must be a  $T$-\container  of all unit line segments, and the corollary follows from Theorem~\ref{thm:pal.kakeya} (Pal's theorem).
\end{proof}

From Corollary~\ref{cor:G2.area.bound},  it suffices to show that any closed curve $\gamma$ in $\ccset$ can be contained in $\triangle_1$ by applying an action of $G_2$ to prove the following main theorem in this section.

\begin{theorem}
\label{thm:G2.main}
The equilateral triangle $\triangle_1$ is the smallest-area convex $G_2$-\container of $\ccset$.
\end{theorem} 

Before proving the theorem, we give some  direct implications of it.
An object $P$ is centrally symmetric if  $-P = P$,  where $-P = \{-x\mid x \in P\}$.  Let $\symset$ be the set of all centrally symmetric closed curves of length $2$. 

\begin{corollary}
\label{cor:equi.tcover.sym}
The equilateral triangle $\triangle_1$ is the smallest-area convex $T$-\container of $\symset$.
\end{corollary}
\begin{proof}
$\symset$ contains all unit segments. Thus from Theorem~\ref{thm:pal.kakeya} (Pal's theorem),
the smallest convex $T$-\container of $\symset$ has area at least 
the area of $\triangle_1$ as in the proof of Corollary~\ref{cor:G2.area.bound}.
On the other hand, $\symset \subset \ccset$, and by Theorem~\ref{thm:G2.main}, any curve in $\symset$ can be contained in $\triangle_1$ by applying a suitable transformation in $G_2$.
However, a centrally symmetric object is stable under the action of $Z_2$, and hence it can be contained in $\triangle_1$ by applying a  transformation in $T$. Thus, $\triangle_1$ is the
smallest-area convex $T$-\container of $\symset$. 
\end{proof}
 
We can consider two special cases shown below. 
Corollary~\ref{cor:equi.tcover.sym} implies that $\triangle_1$ is a $T$-\container of
all rectangles of perimeter $2$, and also of all parallelograms of perimeter $2$.
Since a unit line segment is a degenerate rectangle and a degenerate parallelogram of perimeter $2$, we have the following corollary.
\begin{corollary}
The equilateral triangle $\triangle_1$ is the smallest-area convex $T$-\container of the set of all rectangles of perimeter $2$, and also of the set of all parallelograms of perimeter $2$.
\end{corollary}

We can also obtain the following well-known result about the \containment of 
the set $\wormset$ of all curves of length $1$ (often called {\it worms}) mentioned in the introduction as a corollary.

\begin{corollary} \label{cor:t.container.worm}
The equilateral triangle $\triangle_1$ is the smallest convex $T$-\container of $\wormset$.
\end{corollary}
\begin{proof}
Given a curve $\zeta$ in $\wormset$,  we consider the rotated copy  $-\zeta$ by angle $\pi$.  Suitably translated, we can form a closed curve $\gamma(\zeta)$ by connecting them at their
endpoints.
Then $\gamma(\zeta)\in \symset$, and can be contained in $\triangle_1$ 
under translation.  Therefore, $\zeta$ is also can be contained there. The corollary follows from Theorem~\ref{thm:pal.kakeya} (Pal's theorem).
\end{proof}
We can prove Corollary~\ref{cor:t.container.worm} directly by applying a reflecting argument similar to the proof of Theorem~\ref{thm:G2.main}, which is given in Appendix.

	\begin{figure}[bh]
		\centering
		\includegraphics[scale=.8]{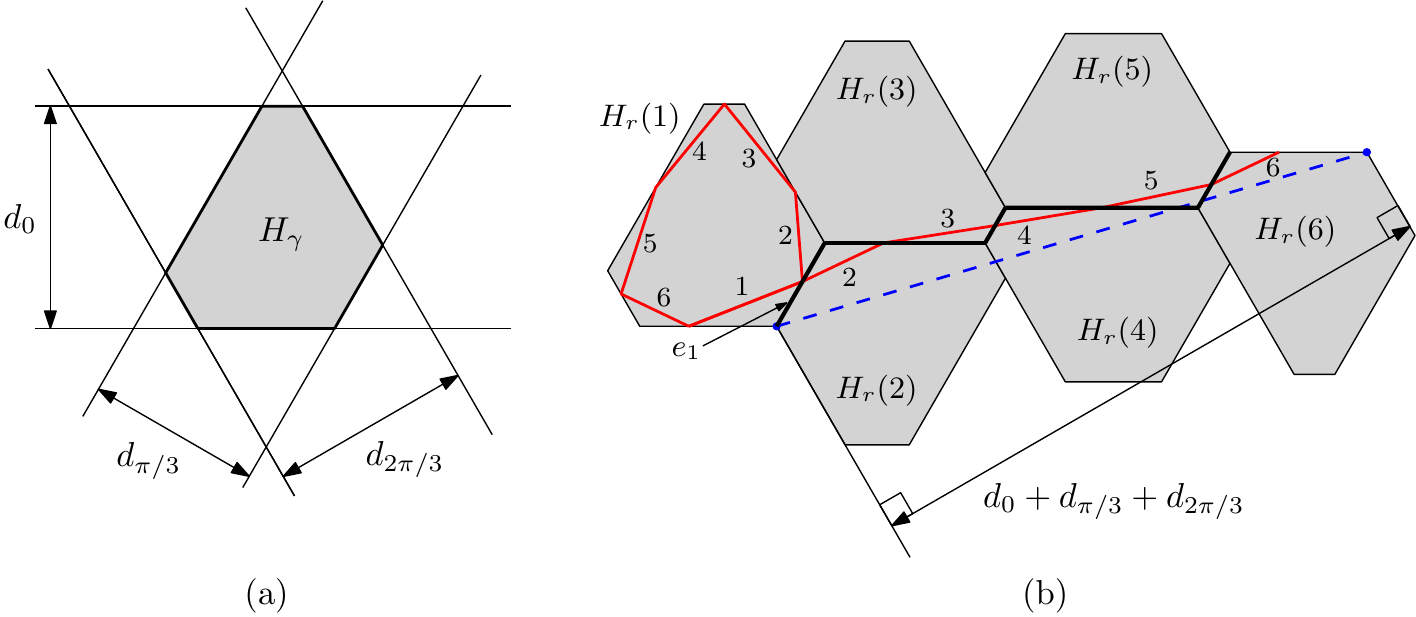}
		\caption{Tiling argument. (a) $H_\gamma=L_0\cap L_{\pi/3}\cap L_{2\pi/3}$. (b) A tiling of six copies $H_\gamma(1),\ldots,H_\gamma(6)$ of $H_\gamma$. Any closed curve that touches every side of $H_\gamma$ is not shorter than the dashed (blue) segment of length at least $d_0+d_{\pi/3}+d_{2\pi/3}$.}
		\label{fig:tiling}
	\end{figure}

\subsection{Proof of Theorem~\ref{thm:G2.main}}
A slab is the region bounded by two parallel lines in the plane, and
its width is the distance between the two bounding lines.
For a closed curve $\gamma$ of $\ccset$, let $L_\theta$ 
be the minimum-width slab of orientation $\theta$ with
$0\leq\theta<\pi$ containing $\gamma$.  We denote the width of $L_\theta$ by $d_\theta$.

\begin{lemma}
	\label{lem:tiling}
	For a closed curve $\gamma$ of $\ccset$, $d_0+d_{\pi/3}+d_{2\pi/3}\leq 2$
	for slabs $L_0$, $L_{\pi/3}$, and $L_{2\pi/3}$ 
	of $\gamma$.
\end{lemma}
\begin{proof}
	Let $H_\gamma$ be the hexagon which is the intersection of $L_0$, $L_{\pi/3}$ and $L_{2\pi/3}$ of $\gamma$. See Figure~\ref{fig:tiling}(a). 
	Let $e_1,\ldots, e_6$ be
	the edges of $H_\gamma$ in counterclockwise order. We can obtain
	a tiling of six copies $H_\gamma(1),\ldots, H_\gamma(6)$ 
	of $H_\gamma$ such that $H_\gamma(1)=H_\gamma$ and $H_\gamma(k+1)$ is the copy of 
	$H_\gamma(k)$ reflected about $e_k$ of $H_\gamma(k)$ 
	for $k$ from $1$ to $5$
	as shown in Figure~\ref{fig:tiling}(b). 
	We observe that the length of a closed curve 
	that touches every side of $H_\gamma$ is at least $d_0+d_{\pi/3}+d_{2\pi/3}$. See Figure~\ref{fig:tiling}(b).
	Since $\gamma$ touches every side of $H_\gamma$ 
	and the length of $\gamma$ is $2$ from the definition of $\ccset$, 
	the lemma follows.
\end{proof}
	\begin{figure}[h]
		\centering
		\includegraphics[scale=.8]{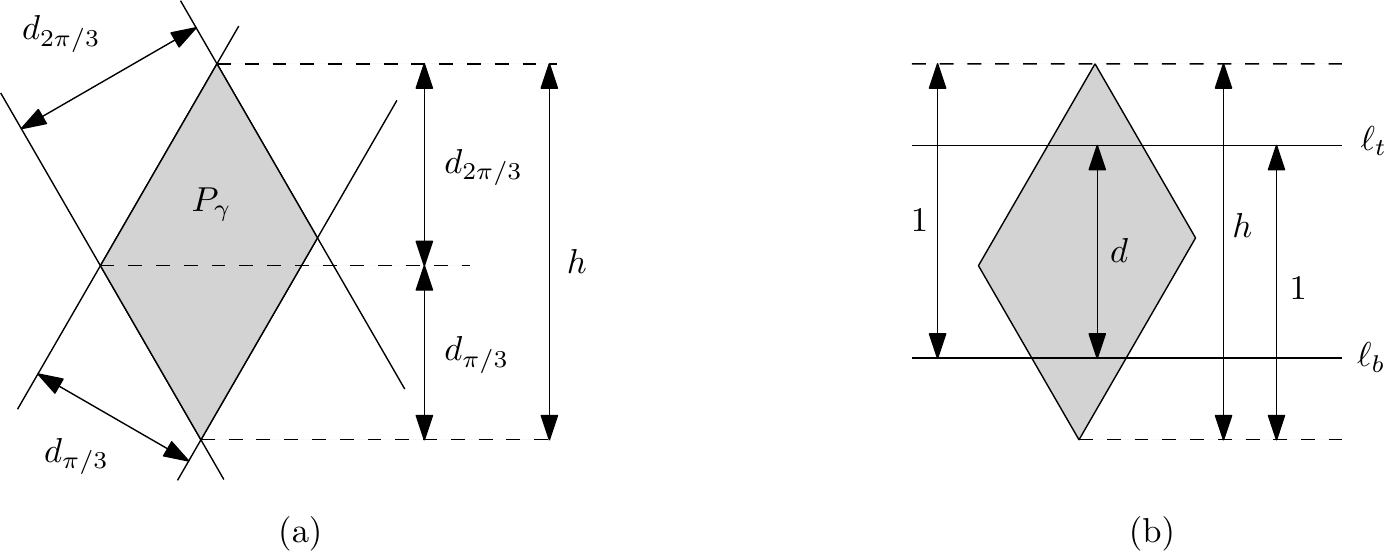}
		\caption{
			(a) $P_\gamma=L_{\pi/3}\cap L_{2\pi/3}$, and $h=d_{\pi/3}+d_{2\pi/3}$.
			(b) $h+d=2$.
		}
		\label{fig:H_containment}
	\end{figure}
\begin{lemma}\label{lem:G2.Sc}
	Any closed curve $\gamma$ of $\ccset$ can be contained in $\triangle_1$ or $-\triangle_1$ under translation. 
\end{lemma}
\begin{proof}
	For a closed curve $\gamma$ of $\ccset$, let $P_\gamma$ be the parallelogram which is the intersection of $L_{\pi/3}$ and $L_{2\pi/3}$ of $\gamma$.
	Then, the height $h$ of $P_\gamma$ is $d_{\pi/3}+d_{2\pi/3}$ (Figure~\ref{fig:H_containment}(a)).
	
	In case of $h\leq1$, $P_\gamma$ can be contained in both $\triangle_1$ and $-\triangle_1$ under translation.
	
	In case of $h>1$, consider two horizontal lines $\ell_t$ and $\ell_b$ such that 
	$\ell_t$ lies above the bottom corner of $P_\gamma$ at distance $1$  
	and $\ell_b$ lies below the top corner of $P_\gamma$ at distance $1$.
For the distance $d$ between $\ell_t$ and $\ell_b$, $d+h$ equals $2$ (Figure~\ref{fig:H_containment}(b)). 
	With Lemma~\ref{lem:tiling}, we have 
	$d_0+d_{\pi/3}+d_{2\pi/3}=d_0+h\leq 2$, so $d\ge d_0$. 
	This implies that $L_0$ does not contain both of the lines $\ell_b$ and $\ell_t$,
	so $L_0$ lies above $\ell_b$ or below $\ell_t$.
	If $L_0$ lies above $\ell_b$, then $\gamma$ can be contained in $\triangle_1$ under translation,
	and if $L_0$ lies below $\ell_t$, then $\gamma$ can be contained in $-\triangle_1$ under translation.
\end{proof}
	
For a closed curve $\gamma$, either $\gamma$ or $-\gamma$ can be
	contained in $\triangle_1$ under translation by Lemma~\ref{lem:G2.Sc}, 
	so we get Theorem~\ref{thm:G2.main}.

\section{Covering under rotation by 90 degrees} 
In this section, we consider the convex $G_4$-\containers of the sets $\segset$ and 
$\ccset$, where $G_4 = T \rtimes Z_4$.  We show that the smallest-area
convex $G_4$-\container of $\segset$ is the isosceles right triangle with legs 
(the two equal-length sides) of length $1$.
For $\ccset$, we give an equilateral triangle of height slightly smaller than $1$ as
a minimal convex $G_4$-\container.

\subsection{Covering of unit line segments}
First, we consider the set $\segset$ of all unit segments.
\begin{theorem}\label{thm:G4.seg}
The smallest-area convex $G_4$-\container of all unit segments 
has area $1/2$, and it is uniquely attained by the isosceles right triangle with legs of length $1$ under closure.
\end{theorem}
\begin{proof}
Let \irtriangle{}\ \ be the isosceles right triangle with legs of length $1$ and 
base of slope $\pi/4$.
Any unit segment of slope $\theta$ can be placed in \irtriangle{}\ \ 
for $0\leq \theta\leq \pi/2$. Thus, \irtriangle{}\ \ 
is a $G_4$-\container of $\segset$, and its area is $1/2$. 

Now we show that any convex $G_4$-\container of $\segset$ 
has area at least $1/2$.
Let $X$ be a smallest-area convex $G_4$-\container of $\segset$, 
and let $s(\theta)$ be a unit segment of slope $\theta$.
Let $A$ be the set of angles $\theta$ such that $s(\theta)$ can be placed in $X$ under translation with $0 \le \theta <\pi$.
Let $A^{-}= A \cap [0, \pi/2)$ and 
$A^{+}=\{\theta-\pi/2\mid \theta\in A\cap [\pi/2, \pi)\}$.
If $A^{-} \cap A^{+} \ne \varnothing$,
then $X$ contains two unit segments which are orthogonal to each other and also contains their convex hull.
Thus, the area of $X$ is at least $1/2$.
So we assume that $A^{-} \cap A^{+} = \varnothing$.
Since $X$ is a $G_4$-\container of $\segset$, $A^{-} \cup A^{+} = [0, \pi/2)$.
If $A^{-} = \varnothing$, then $A^{+}=[0,\pi/2)$ and
there is a sequence $\{\theta_k\}_{k=1}^{\infty}$ in $A^{+}$ 
such that $\lim_{k\to\infty}\theta_k=\pi/2$.
Since $X$ contains a unit segment $s(\pi/2)$ and a unit segment $s(\theta_k+\pi/2)$
for any $k$, $X$ contains their convex hull.
Since $\lim_{k\to\infty}\theta_k=\pi/2$,
the area of $X$ is at least $1/2$.
Similarly, we can prove this for the case $A^{+} = \varnothing$.
So we assume that $A^{-} \ne \varnothing$ and $A^{+} \ne \varnothing$.
Then there are two cases : $A^{-}$ contains no interval or 
$A^{-}$ contains an interval.

Consider the case that $A^{-}$ contains no interval.
Let $I(x, \epsilon)$ be the open interval in $[0,\pi/2)$ 
centered at an angle $x$ with radius 
$\epsilon$, and let $\theta'$ be an angle in $A^{-}$.
Since $A^{-}$ contains no interval, 
there is a sequence $\{\theta'_k\}_{k=1}^{\infty}$ such that $\theta'_k \in I(\theta', 1/k) \cap A^+$ for each $k$, and $\lim_{k\to\infty}\theta'_k=\theta'$.
Since $X$ contains a unit segment $s(\theta')$ and a unit segment 
$s(\theta'_k+\pi/2)$ for any $k$, 
$X$ contains their convex hull.
Thus, the area of $X$ is at least $1/2$.

Consider now the case that $A^{-}$ contains an interval.
Since $A^{+} \ne \varnothing$, there is an interval in $A^-$
with an endpoint $\bar{\theta}$, other than $0$ and $\pi/2$.
Then there is a sequence $\{\bar{\theta}_k\}_{k=1}^{\infty}$ in $A^-$ 
such that $\lim_{k\to\infty}\bar{\theta}_k=\bar{\theta}$. 
Since $\bar{\theta}$ is an endpoint of the interval in $A^-$,
$I(\bar{\theta}, 1/n) \cap A^+\ne\varnothing$ for any $n$.
So there is a sequence $\{\bar{\theta}'_n\}_{n=1}^\infty$
such that $\bar{\theta}'_n \in I(\bar{\theta}, 1/n) \cap A^+$ for each $n$,
and $\lim_{n\to\infty}\bar{\theta}'_n=\bar{\theta}$.
Since $X$ contains unit segments $s(\bar{\theta}_k)$ and $s(\bar{\theta}'_n+\pi/2)$ for any $k$ and any $n$, 
$X$ contains their convex hull.
Thus, the area of $X$ is at least $1/2$.

We show the uniqueness of the smallest-area convex $G_4$-\container of $\segset$ under its closure. We assume that $X$ is compact.
First, we prove that if there is a sequence of angles in $A$ 
converging to $\theta\in[0,\pi)$, then $X$ contains a unit segment $s(\theta)$.
Consider a sequence $\{\theta_k\}_{k=1}^\infty$ in $A$ such that 
$\lim_{k\to\infty}\theta_k=\theta$.
Then there is a sequence $\{s(\theta_k)\}_{k=1}^\infty$ of unit segments 
$s(\theta_k)$ contained in $X$.
Since $X$ is compact,
there is a subsequence  $\{s(\theta_{k_p})\}_{p=1}^\infty$ 
that converges to a unit segment $s(\theta)$.
Since $X$ is closed, $s(\theta)$ is contained in $X$.
By the argument in the previous paragraphs, 
there is an angle $\tilde{\theta} \in [0, \pi/2)$ 
such that $X$ contains unit segments $s(\tilde{\theta})$ and 
$s(\tilde{\theta}+\pi/2)$, orthogonal to each other.
Moreover, if $X$ is  a \container of area $1/2$, 
$X$ is the convex hull $\chull$ of $s(\tilde{\theta})$ and $s(\tilde{\theta}+\pi/2)$.
There are two cases of $\chull$ : it is either a convex quadrilateral or a triangle of height $1$ and base length $1$.  

Consider the case that $\chull$ is a convex quadrilateral. 
Then both $\tilde{\theta}$ and $\tilde{\theta}+\pi/2$ are 
isolated points in $A$.
Thus, $\chull$ is not a $G_4$-\container of $\segset$.
Consider the case that $\chull$ is a triangle of height $1$ and base length $1$. 
Without loss of generality, assume that $s(\tilde{\theta})$ 
is the base of $\chull$ and $s(\tilde{\theta}+\pi/2)$ 
corresponds to the height of $\chull$.
Suppose that $\chull$ is an acute triangle.
Since $s(\tilde{\theta})$ is the base of $\chull$, 
$\tilde{\theta}$ is an isolated point in $A$.
Observe that $s(\tilde{\theta}+\pi/2)$ can be rotated infinitesimally 
around the corner opposite to the base in both directions 
while it is still contained in $\chull$.
Thus, $\tilde{\theta}+\pi/2$ is an interior point of an interval $I$ in $A$.
So for an endpoint $\hat{\theta}$ of $I$ other than $\tilde{\theta}$,
$\chull$ contains unit segments $s(\hat{\theta})$ and $s(\hat{\theta}+ \pi/2)$.
The convex hull $\chull'$ of the two segments is also a convex quadrilateral or
a non-obtuse triangle of area at least $1/2$. 
Since $\chull'$ is contained in $\chull$ of area $1/2$, $\chull'=\chull$,
that is, $\chull'$ must be an acute triangle with base $b$ of length 1.
Since $\hat{\theta}\neq\tilde{\theta}$, $b$ is not the base of $\chull$, 
and $b$ must be one of the other two sides of $\chull$.
But clearly both the sides are longer than $1$, the height of $\chull$, 
a contradiction. Hence $\chull$ is not an acute triangle. 
Thus, $\chull$ is the isosceles right triangle, and it is the unique 
$G_4$-\container of $\segset$ under closure.
\end{proof}

\subsection{ An equilateral triangle \container of closed curves}
Unfortunately, the isosceles right triangle with legs of length $1$
is not a  $G_4$-\containers of $\ccset$, since it does not contain a circle 
of perimeter $2$. 
Using the convexity of the area function on the convex hull of two translated convex objects~\cite{Ahn2012}, 
it can be shown that 
the area of the convex hull of  the union of the isosceles  right triangle and any translated copy of the disk  has area at least 0.543. 

Naturally,  the equilateral triangle of height $1$ is a $G_4$-\container of $\ccset$, since it is a $G_2$-\container of $\ccset$.
However, a smaller equilateral triangle can be a $G_4$-\container of $\ccset$, and 
we seek for the smallest one, which is conjectured to be the smallest-area 
$G_4$-\container of $\ccset$. 
Consider an equilateral triangle that is a $G_4$-\container of $\ccset$.  Since it is a $G_4$-\container of  $\segset$, it must contain a pair of orthogonally crossing unit segments as shown in the proof of Theorem~\ref{thm:G4.seg}.
Figure~\ref{fig:G4-tri-beta-a} shows a smallest equilateral triangle containing such 
a pair.
	
The side length of the triangle is $ \frac{1 + \sqrt{3}}{\sqrt{6}} \approx 1.115$ 
and the height is $\beta = \cos (\pi /12) = \frac{1+ \sqrt{3}}{2 \sqrt{2}}  \approx 0.966$. 
Its area is $\frac{2\sqrt{3} +3 }{12} \approx  0.538675$. 
We denote this triangle by $\triangle_{\beta}$.  

\begin{figure}[ht]
		\centering
		\includegraphics[scale=.9]{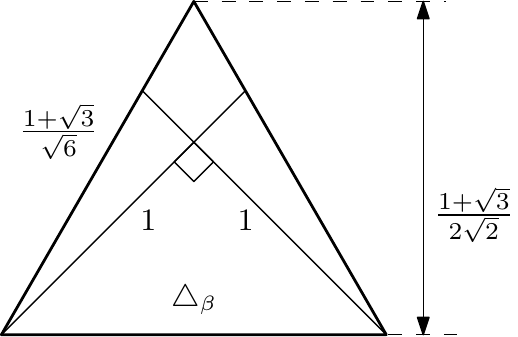}
		\caption{ The smallest equilateral triangle $\Delta_\beta$ containing a pair of orthogonal unit segments.
		}
		\label{fig:G4-tri-beta-a}
	\end{figure}

\begin{theorem}\label{thm:G4.main}
The equilateral triangle $\triangle_{\beta}$ is a convex $G_4$-\container of 
all closed curves of length $2$.
\end{theorem}

\begin{definition}
\label{def:brimful}
A closed curve $\gamma$ is called $G$-\brimful for a set $\Lambda$ 
if it can be placed in $\Lambda$
but cannot be placed in any smaller scaled copy of $\Lambda$
under $G$ transformation.
\end{definition}

To prove Theorem~\ref{thm:G4.main}, we consider the following claim.

\begin{claim}
The length of any $G_4$-\brimful curve for $\triangle_{\beta}$ is at least $2$.
\end{claim}

The claim implies that any closed curve of length smaller than $2$ 
can be placed in $\triangle_{\beta}$ under $G_4$ transformation, 
but it is not $G_4$-\brimful for $\triangle_{\beta}$.
Thus, if the claim is true, then
every closed curve of length 2 can be placed in $\triangle_{\beta}$
under $G_4$ transformation, so Theorem~\ref{thm:G4.main} holds.
Thus, we will prove the claim. 

Let $\gamma$ be a $G_4$-\brimful curve for $\triangle_{\beta}$.
Let 
$g^i =e^{2i\pi \sqrt{-1} /4}\in Z_4$ for $i=0,\ldots,3$.
We consider the rotated copy 
$g^i \triangle_\beta$ of $\triangle_\beta$,
and find the smallest scaled copy $\Upsilon_i$ of $g^i \triangle_\beta$
which circumscribes $\gamma$ with a proper translation.
Let $a_i$ be the scaling factor, and the \brimful property implies that $a_i \ge 1$ for $i=0,1,2,3$ and the minimum of them equals $1$.
 Also, $\gamma$ touches all three sides of $\Upsilon_i$ for each $i$, where we allow it 
 to touch two sides simultaneously at the vertex shared by them.

	As illustrated in Figure~\ref{fig:G4-peri-same}(a), the intersection $H=H(\gamma) = \bigcap_{i=0}^3 \Upsilon_i$ contains $\gamma$, and 
   $H$ is a (possibly degenerate) 
	convex $12$-gon such that $\gamma$ touches all 12 edges of $H$.
The $12$-gon can be degenerate as shown in Figure~\ref{fig:G4-peri-same}(b), where  $\gamma$ touches multiple edges simultaneously at vertices 
corresponding to degenerated edges.  
Note that $H$ consists of edges of slopes $2k\pi/12\; (\bmod\; \pi)$ 
for $k=0,\ldots,11$.

The following lemma states that the perimeter of $H$ only depends on $a_i$ $(i=0,1,2,3)$.
\begin{figure}[hb]
	\centering
	\includegraphics[scale=.8]{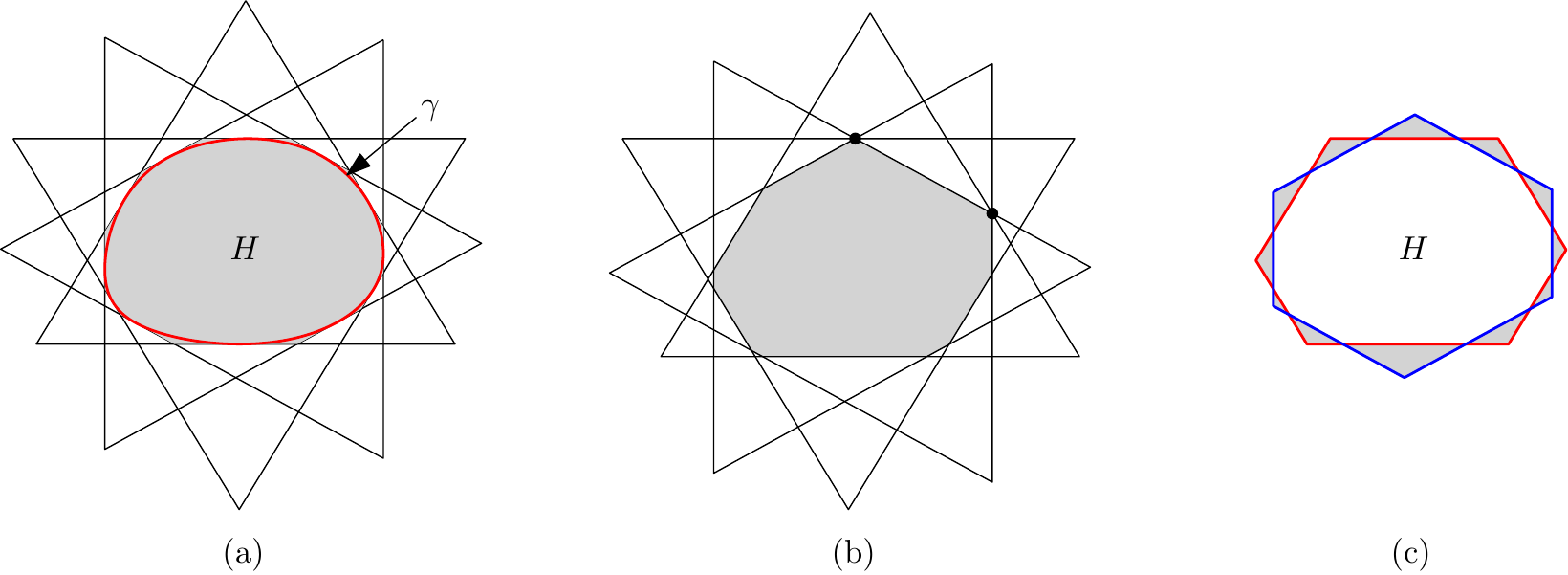}
	\caption{
		(a) $H=\bigcap_{i=0}^3 \Upsilon_i$ is a (possibly degenerate) convex $12$-gon.
		$\gamma$ touches every edge of $H$.
		(b) An example of a degenerated convex $12$-gon. The degenerated edges 
		are denoted by black disks.
		(c) For $X_{0,2}=\Upsilon_0 \cap \Upsilon_2$ (red) and $X_{1,3}=\Upsilon_1 \cap \Upsilon_3$ (blue), $(X_{0,2} \cup X_{1,3}) \setminus H$ consists of 12 triangles (gray), each of which 
is an isosceles triangle with apex angle $2\pi/3$.
	}
	\label{fig:G4-peri-same}
\end{figure}
\begin{lemma}\label{lem:hexagon.perimeter.bound}
The perimeter $\peri{H}$ equals $\frac{a_0+a_1 + a_ 2 + a_3}{2\cos(\pi/12)}$.
\end{lemma}
\begin{proof}

Let $\delta$ be the side length of $\Delta_{\beta}$. 
The intersection $X_{0,2}= \Upsilon_0 \cap \Upsilon_2$ is a hexagon where $\gamma$ touches all of its six (possibly degenerated) edges.
Then, $(\Upsilon_0 \cup \Upsilon_2) \setminus X_{0,2}$ consists of six equilateral triangles, and the total sum of the perimeters of them equals 
$\peri{\Upsilon_0} + \peri{\Upsilon_2} = 3(a_0 + a_2) \delta$.  
On the other hand, one edge of each equilateral triangle contributes to
 the  boundary polygon of $X_{0,2}$.  Hence, $\peri{X_{0,2}} = (a_0 + a_2)\delta$.  Similarly, the perimeter of $X_{1,3} = \Upsilon_1 \cap \Upsilon_3$ equals $(a_1 + a_3)\delta$. 

Now, consider $H = \bigcap_{i=0}^3 \Upsilon_i = X_{0,2} \cap X_{1,3}$. See Figure~\ref{fig:G4-peri-same}(c).
The set $(X_{0,2} \cup X_{1,3}) \setminus H$ consists of 12 triangles, each of which 
is an isosceles triangle with apex angle $2\pi/3$.   The total sum of the perimeters of these triangles equal $\peri{X_{0,2}} + \peri{X_{1,3}}$, 
while the bottom side of each isosceles triangle contributes to the 12-gon $H$. 
Since the ratio of the bottom side length to the perimeter of the isosceles triangle is $\frac{\sqrt{3}}{2+ \sqrt{3}}$, 
\begin{eqnarray*}
\peri{H} &=& \frac{\sqrt{3}}{2 + \sqrt{3}} (\peri{X_{0,2}} + \peri{X_{1,3}})\\
 &=& \frac{\sqrt{3}}{2+ \sqrt{3}}(a_0+a_1+a_2+a_3) \delta\\
  &=& \frac{2}{2+\sqrt{3}} \cos(\pi/12)(a_0+a_1+a_2 + a_3)\\
  &=& \frac{a_0+a_1+a_2+ a_3} {2\cos(\pi/12)}
 \end{eqnarray*} 
The third equality comes from $\delta = \frac{2}{\sqrt{3}} \beta = \frac{2}{\sqrt{3}}\cos(\pi /12)$, 
and the last equality comes from 
$\cos(\pi/12) =  \frac{1+ \sqrt{3}}{2 \sqrt{2}} $ and hence $\cos^2(\pi/12) = \frac{2+\sqrt{3}}{4}$. 
\end{proof}

The following lemma states the relation between the lengths of $\gamma$ and $H$.
 \begin{figure}[t]
		\centering
		\includegraphics[scale=.8]{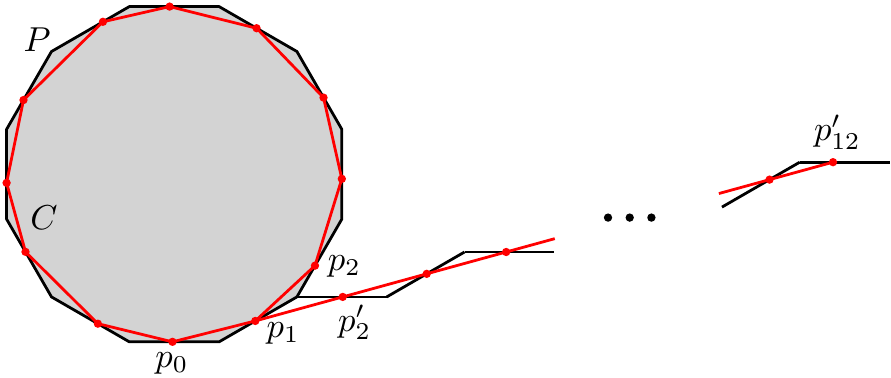}
		\caption{Illustration of the proof of Lemma~\ref{lem:G4.12gon}.
		}
		\label{fig:G4-tri-beta-b}
	\end{figure}
\begin{lemma}\label{lem:G4.12gon}
Let $P$ be a (possibly degenerate) convex $12$-gon with edges 
$e_k$ of slopes $2k\pi/12\; (\bmod\; \pi)$ 
for $k=0,\ldots,11$.
Let $C$ be a circuit that connects $12$ points, one point on each edge of $P$, 
in order along the boundary of $P$.
Then $\peri{C} \ge \peri{P} \cos (\pi/12)$. 
\end{lemma}
\begin{proof}
Let $p_k$ be the point on $e_k$ for $k=0,\ldots,11$.
The circuit $C$ connects the points from $p_0$ to $p_{12}$ in increasing order of their
indices, where
	$p_{12}$ is a duplicate of $p_0$. We also assume that $p_{12}$ 
	is on $e_{12}$ which is a duplicate of $e_0$.
While incrementing $k$ by $1$ from $1$ to $11$,
	we reflect the edges $e_{k+1},\ldots, e_{12}$
	about the edge $e_k$.
Then, the edges $e_0$,\ldots,$e_{12}$ are transformed to 
a zigzag path alternating a horizontal edge and an edge of slope $2\pi/12$, 
and $C$ becomes the path 
connecting $p_0, p_1, p'_2,\ldots,p'_{12}$, where $p'_k$ is the location
of $p_k$ after the series of reflections for $k=2,\ldots,12$.
Thus, $\peri{C}$ is at least the distance between $p_0$ to $p'_{12}$.
See Figure~\ref{fig:G4-tri-beta-b} for the illustration.
The vector $x = p_0 p'_{12}$ is the sum of a horizontal vector $a$ and 
another vector $b$ of slope $\pi/6$ such that $|a| + |b| = \peri{P}$.
The size $|x|$ of $x$ is minimized when $|a|=|b|$
to attain the value $(|a|+|b|) \cos(\pi/12)$. Thus, $|x| \ge \ell(P) \cos(\pi/12)$
and $\peri{C} \ge \peri{P} \cos(\pi/12)$.
\end{proof}

From Lemma~\ref{lem:G4.12gon}, the following corollary is immediate. 

\begin{corollary}\label{cor:perimeter}
For any $G_4$-\brimful curve $\gamma$ for $\triangle_{\beta}$, 
$\peri{\gamma} \ge \peri{H} \cos (\pi/12)$,  $H = H(\gamma)$. 
\end{corollary}

Combining Lemma~\ref{lem:hexagon.perimeter.bound} and Corollary~\ref{cor:perimeter}, 
we have 
$\peri{\gamma} \ge \peri{H} \cos (\pi/12) = (a_0+ a_1 + a_2+ a_3)/2 \ge 2$. The claim is proved.

\subsection{Minimality  of the \container}
\begin{figure}[b]
	\centering
	\includegraphics[scale=.8]{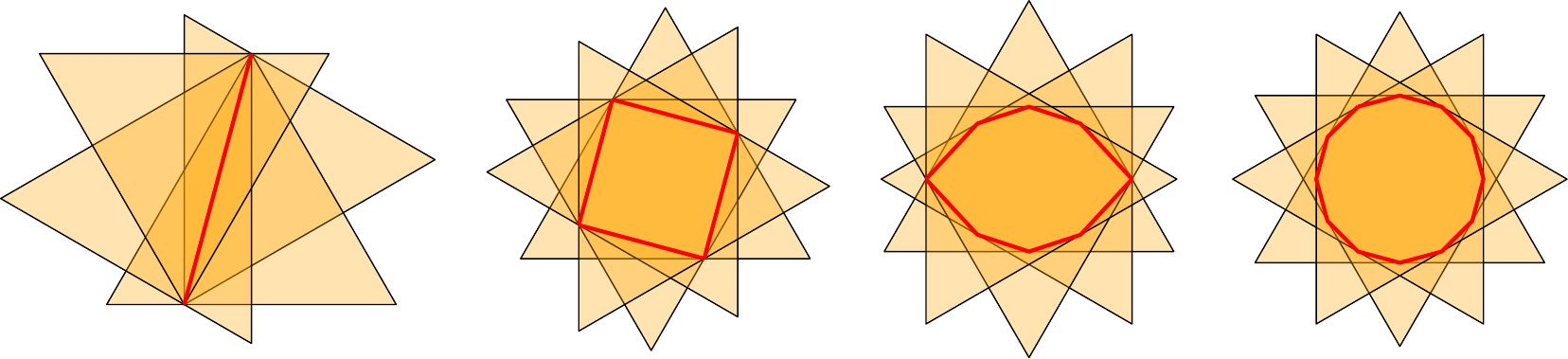}
	\caption{Examples of $G_4$-\brimful curves of length 2.}
	\label{fig:G4-brimful}
\end{figure}

There are many $G_4$-\brimful curves of length $2$, and some examples are  illustrated as red curves (the line segment is considered as a degenerate closed curve) in Figure~\ref{fig:G4-brimful} together with the (possibly degenerate) $12$-gons they inscribe.
They suggest that $\triangle_{\beta}$ is the smallest-area convex $G_4$-\container of $\ccset$. 

\begin{conjecture}
$\triangle_{\beta}$ is the smallest-area convex $G_4$-\container of $\ccset$.
\end{conjecture}

 Although we have not yet proven it rigorously, we can prove that it is minimal in the 
set theoretic sense.  That is, no proper closed subset of $\triangle_{\beta}$ can be a 
$G_4$-\container of $\ccset$.

Without loss of generality, assume that $\triangle_{\beta}$ is located with 
its bottom side being horizontal.
Let $S$ be the set of unit line segments that can be contained in $\triangle_{\beta}$ 
under translation (without considering the $Z_4$-action).
Then, the set of slope angles of the line segments of $S$ is $A=[0, \pi/12] \cup [3\pi/12, 5\pi/12] \cup [7\pi/12,9\pi/12] \cup [11\pi/12, \pi]$.   Let $A' = (0, \pi/12) \cup (3\pi/12, 5\pi/12) \cup (7\pi/12,9\pi/12) \cup (11\pi/12, \pi)$ and let $S_{A'}$ be the set of all unit line segments 
of slopes in $A'$.

\begin{lemma}\label{lem:covering_Sa'}
$\triangle_{\beta}$ is the smallest-area closed convex $T$-\container of $S_{A'}$. 
\end{lemma}
\begin{proof}
Consider six unit segments of slopes $\pi/12+i\cdot\pi/6$ for $i=0,1,\ldots,5$.
Then any compact convex $T$-\container of $S_{A'}$ contains 
these six unit segments under translation.

Ahn \etal~\cite{Ahn2014} proved there exists a triangle that is the smallest-area 
convex $T$-\container of any given set of segments, 
and showed an algorithm to construct the triangle. 
For the six unit segments,
their algorithm computes the smallest regular hexagon that contains 
them under translation. Since the hexagon is the Minkowski symmetrization of $\triangle_{\beta}$, which is $\{\frac{1}{2}(x-y)\mid x,y\in\triangle_{\beta} \}$, 
the algorithm returns $\triangle_{\beta} $ as 
the smallest-area closed convex $T$-\container of the six unit segments.
Thus, the lemma follows.
\end{proof}

\begin{proposition}
$\triangle_{\beta}$ is a minimal closed convex $G_4$-\container of $\ccset$.
\end{proposition}
\begin{proof}
Suppose $P \subseteq \triangle_{\beta}$ is a closed subset and a $G_4$-\container of $\ccset$. 
Observe that each angle $\theta$ in $A'$ has no other angle in 
$A$ that is equivalent to $\theta$ under the $Z_4$ action.
Thus, $P$ must contain all line segments of $S_{A'}$ under translation. 
Lemma~\ref{lem:covering_Sa'} implies that the area of $P$ is the same as that of $\triangle_{\beta}$.
Therefore, $P$ must be $\triangle_{\beta}$ itself.
\end{proof}

\section{\texorpdfstring{$G_k$}{Gk}-\container of unit line segments}
Consider the smallest-area convex $G_k$-\container of 
the set $\segset$ of all unit line segments.
In general, a smallest-area convex $T$-\container of 
any given set of segments is attained by a triangle~\cite{Ahn2014}.
This implies that there is a triangle that is a smallest-area convex $G_k$-\container
of $\segset$.
The following theorem determines the set of all smallest-area convex $G_k$-\containers.

\begin{theorem}
If $k \ge 3$ is odd, the smallest area of convex $G_k$-\container of $\segset$ is $\frac{1}{2} \sin (\pi /k) $, and it is attained 
	by any triangle $\triangle XYZ$ with bottom side $XY$ of length $1$ and
	height $\sin( \pi/ k)$ such that $\pi/2 \le \angle X \le (k-1)\pi/k$. 
	If $k \ge 4 $ is even, the smallest area of convex $G_k$-\container of $\segset$ is $\frac{1}{2} \sin (2\pi/k) $, and it is attained 
	by any triangle $\triangle XYZ$ with bottom side $XY$ of length $1$ and height 
	$\sin (2\pi/ k)$ such that $ \pi/2 \le \angle X \le (k-2)\pi / k$. 
	\end{theorem}
\begin{proof}
We have already seen a proof for $k=4$, and it can be generalized as follows.
Let $\Lambda$ be a smallest-area convex $G_k$-\container of $\segset$. 

First, let us consider the case that $k$ is odd.
Since $Z_k$ consists of $2i\pi/k$ rotations for $i=0,1,2,\ldots,k-1$, 
one of the unit segments of slopes $\theta+ 2i\pi/k\; (\bmod\; \pi )$ 
must be contained in $\Lambda$ 
under translation for each angle $\theta$ with $0 \le \theta < 2 \pi/k$.
We define the smallest one of such  angles to be $f(\theta)$. 
As before, we can assume there exists an angle $\bar{\theta}$
such that $f(\bar{\theta}) = \bar{\theta}$.

Let $A = \{ f(\theta) \mid  0 \le \theta < 2\pi/k \} $ be the set of angles, let
$\bar{A}$ be the complement of $A$, and
let $s(\theta)$ be a unit segment of slope $\theta$.
There  exists an angle $\tilde{\theta}$ with $0 \le \tilde{\theta} <  2\pi/k $ such that 
$\tilde{\theta}$ is contained in both $A$ and the closure of $\bar{A}$.
There is a sequence $\{\tilde{\theta}_n\}_{n=1}^\infty \subseteq A$ 
such that $\lim_{n \to \infty} \tilde{\theta}_n = \tilde{\theta}+2i\pi/k$ for some $i$.
Then, $\Lambda$ contains two segments $s(\tilde{\theta})$ and $s(\tilde{\theta}_n)$ for any $n$, 
and their convex hull $\chull_n$.
Since $\lim_{n \to \infty} \tilde{\theta}_n = \tilde{\theta}+2i\pi/k$,
$\lim_{n \to \infty}\lVert \chull_n \rVert = \frac{1}{2} |\sin (2 i \pi / k )|$.
Since $\frac{1}{2} |\sin (2 i \pi / k )|$ is minimized at $i= (k-1)/2$  , the area of $\Lambda$ is at least $\frac{1}{2} \sin (\pi /k)$. 

On the other hand,
let $\triangle XYZ$ be a triangle with bottom side $XY$ of length $1$ and
height $\sin( \pi/ k)$ such that $\pi/2 \le \angle X \le (k-1)\pi/k$. Then $\angle Y +\angle Z\ge \pi/k$. We show that any segment of slope $\theta$ with $0\le\theta\le \pi/k$ can be placed
within $\triangle XYZ$. 
Any unit line segment of slope $\theta_1$ for $0\le\theta_1\le\angle Y$ can be placed within $\triangle XYZ$
with one endpoint at $Y$, and any unit line segment of slope $\theta_2$ for $\angle Y \le\theta_2\le \max\{\angle Y, \pi/k\}$ can be placed within $\triangle XYZ$ with one endpoint at $Z$. 
Thus, any segment of slope $\theta$ with $0\le\theta\le \pi/k$ can be placed
within $\triangle XYZ$.

The case of even $k$ can be proven analogously. The only difference is  that $i \ne k/2$, 
and the minimum area is attained at $i = k/2 -1$.
Analogously to the odd $k$ case, the area can be minimum only if the convex hull of $s(\theta) \cup s(\theta + 2i \pi/k)$ is a triangle, and 
it is routine to derive the 
conditions that the triangle is a $G_k$-\container.
\end{proof}

\section{Covering under rotation by 120 degrees}
\label{sec:contain.g3}
We construct a convex $G_3$-\container of $\ccset$, denoted by $\Gamma_3$, 
as follows. Let $\Gamma$ be the convex region bounded by 
$y^2=1+2x$ and $y^2=1-2x$, and containing the origin $O$.
Then $\Gamma_3$ is the convex subregion of $\Gamma$ bounded by 
the $x$-axis and the line $y=2/3$. 
The area of $\Gamma_3$ is $|\Gamma_3| = 2 \left(\frac{5}{27}+\int_{\frac{5}{18}}^{\frac{1}{2}} \sqrt{1-2x}\, dx \right)=\frac{46}{81}$, which is smaller than $0.5680$.
See Figure~\ref{fig:G3-curve-properties}(a).
\begin{figure}[t]
	\centering
	\includegraphics[scale=0.8]{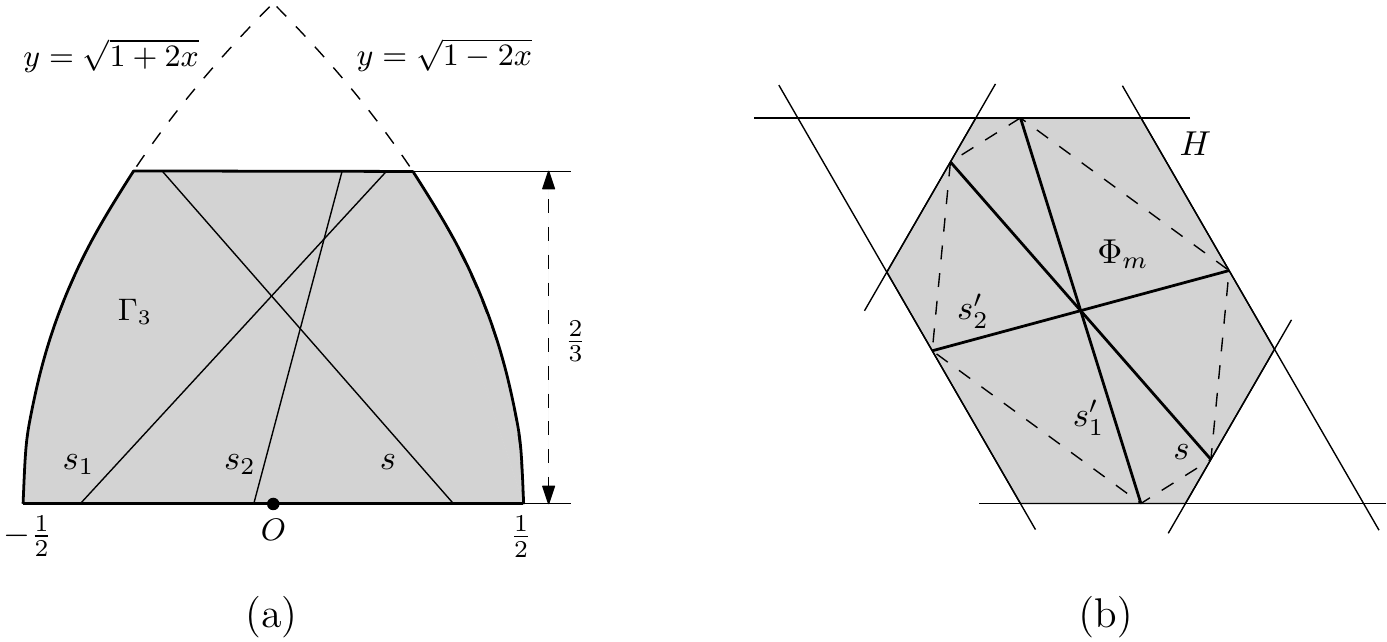}
	\caption{(a) Illustration of $\Gamma_3$. (b) $H=L_0\cap L_{\pi/3}\cap L_{2\pi/3}$.
	}
	\label{fig:G3-curve-properties}
\end{figure}

We show that $\Gamma_3$ is a $G_3$-\container of $\ccset$.
We first show a few properties that we use in the course.
Let $\Gamma^+$ be the region of $\Gamma$ 
above the $x$-axis. 
We call the boundary segment on the $x$-axis the \emph{bottom side},
the boundary curve on $y=\sqrt{1+2x}$ the \emph{left side}, and the boundary curve
on $y=\sqrt{1-2x}$ the \emph{right side} of $\Gamma^+$.
$\Gamma^+$ is called a {\it church window}, which is a $T$-\container of 
$\ccset$~\cite{BC89}.
 	The following lemma gives a lower bound on the length of a
 	$T$-\brimful curve for $\Gamma^+$.
\begin{lemma} \label{lem:G3.curve.length}
Any closed $T$-\brimful curve for $\Gamma^+$ has length at least 2.
\end{lemma}
\begin{proof}
Recall the definition of a $G$-\brimful curve in Definition~\ref{def:brimful}.
Consider a closed $T$-\brimful curve $\gamma$ of minimum length for
$\Gamma^+$. Observe that $\gamma$ touches every side of $\Gamma^+$;
otherwise $\gamma$ can always be translated to lie in the interior 
of $\Gamma^+$.
Let $\triangle{XYZ}$ be a triangle
for the touching points $X, Y, Z$ of the curve with the boundary
of $\Gamma^+$. 
Since $\peri{\gamma}\ge\peri{\triangle{XYZ}}$,	$\gamma$ is $\triangle{XYZ}$.
\begin{figure}[b]
	\centering
	\includegraphics[scale=0.85]{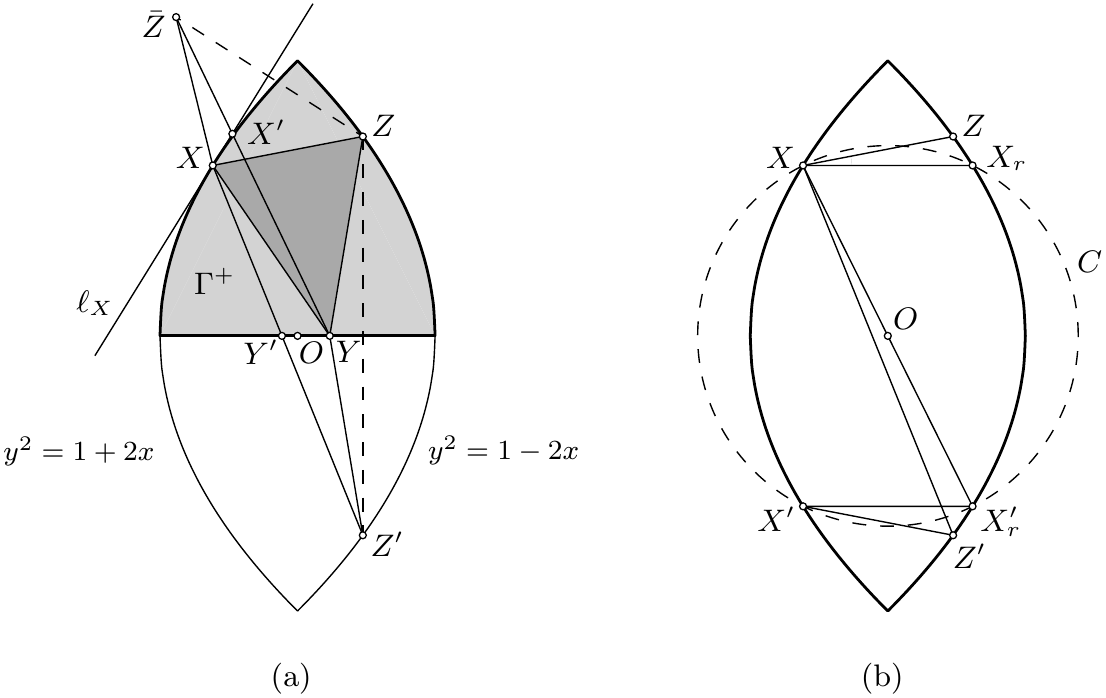}
	\caption{
		Illustration of the proof of Lemma~\ref{lem:G3.curve.length}.
	}
	\label{fig:G3-curve-length}
\end{figure}

Without loss of generality,
assume that $X$ is on the left side, $Z$ is on the right side,
and $Y$ is on the bottom side of $\Gamma^+$.
If $X$ and $Y$ are at $(-1/2, 0)$, or $X$ and $Z$ are at $(0, 1)$,
or $Y$ and $Z$ are at $(1/2,0)$, 
$\gamma$ becomes a line segment 
and the length of the line segment is always larger than or equal to $1$. 
Thus, $\peri{\gamma}\geq 2$.

Now we assume that none of $X,Y$, and $Z$ is on a corner of $\Gamma^+$.
Let $\ell_X$ be the line tangent to the left side of $\Gamma^+$ at $X$.
Let $Z'$ be the point symmetric to $Z$ with respect to the $x$-axis,
and let $\bar{Z}$ be the point symmetric to $Z$ with respect to $\ell_X$.
See Figure~\ref{fig:G3-curve-length}(a).

We claim that $\bar{Z}, X, Y$, and $Z'$ are collinear.
If $X, Y$, and $Z'$ are not collinear, consider the intersection point $Y'$ 
of $XZ'$ and the bottom side of $\Gamma^+$. 
Then, $\peri{\triangle{XY'Z}}<\peri{\triangle{XYZ}}$, since 
$|XY'| + |Y'Z| = |XZ'| < |XY|+|YZ'|=|XY|+|YZ|$. 
This contradicts the assumption on $\gamma$.
Thus, $X, Y$, and $Z'$ are collinear.
If $\bar{Z}$ is not on the line through $X$ and $Y$, 
consider the point $X'$ where $Y \bar{Z}$ intersects $\ell_X$.
Then we have $\peri{\triangle{X'YZ}}<\peri{\triangle{XYZ}}$ because
$|YX'| + |X'Z| = |Y\bar{Z}| < |YX|+|X\bar{Z}|=|YX|+|XZ|$.
For the point $X''$ where $X'Y$ intersects the left side of $\Gamma^+$,
we have $\triangle{X''YZ}\subset \triangle{X'YZ}$. 
Therefore, $\peri{\triangle{X''YZ}}<\peri{\triangle{X'YZ}}<\peri{\gamma}$, and
this contradicts the assumption on $\gamma$.

Suppose that $XZ$ is not horizontal.
From the collinearity of $X, \bar{Z},$ and $Z'$, 
the reflection of $\ell_{XZ}$ in the tangent line $\ell_X$
is $\ell_{XZ'}$. 
Let $X_r$ be the point symmetric to $X$ with respect to the $y$-axis, and
let $X'$ and $X'_r$ be the points symmetric to $X$ and $X_r$ with respect to the $x$-axis.
By the symmetry, $\angle ZXX_r =\angle Z'X'X'_r$.
From the geometry of the parabola, 
the reflection of $\ell_{XX_r}$ in the tangent line $\ell_X$
is $\ell_{XX'_r}$. 
Therefore, $\angle Z'X'X'_r=\angle Z'XX'_r$, implying that 
$X, X', X'_r$ and $Z'$ are on a circle $C$.
See Figure~\ref{fig:G3-curve-length}(b).
Since $C$ passes through $X, X'$ and $X'_r$, the center of $C$ 
is at the origin. This implies that $X'Z'$ is horizontal, and thus
$XZ$ is also horizontal. This contradicts that $XZ$ is not horizontal.

Since $XZ$ is horizontal, $Y$ is at the origin.
Thus, $\triangle{XYZ}$ is an isosceles triangle with base $XZ$,
and $\peri{\triangle{XYZ}}=2$ by the construction of $\Gamma^+$.
\end{proof}

The following lemma shows the convexity of the perimeter function
on the convex hull of planar figures under translation.
\begin{lemma}[Theorem 2 of~\cite{Ahn2012}]
\label{lem:peri.convex}
For $k$ compact convex figures $C_i$ for $i=1,\ldots,k$ in the plane,
the perimeter function of their convex hull of $C_i(r)$ is convex,
where $C_i(r)$ for a vector $r=(r_1,\ldots,r_k)\in\mathbb{R}^{2k}$ is $C_i+r_i$
for $r_i\in\mathbb{R}^2$.
\end{lemma}

For compact convex figures that have point symmetry in the plane, 
we can show an optimal translation of them 
using the convexity of the perimeter function in Lemma~\ref{lem:peri.convex}.
\begin{lemma} \label{lem:k.segments.convexhull}
For $k$ compact convex figures $C_i$ for $i=1,\ldots,k$ that have 
point symmetry in the plane, the perimeter function of 
their convex hull of $C_i(r)$ is minimized
when their centers (of the symmetry) meet at a point, 
where $C_i(r)$ for a vector $r=(r_1,\ldots,r_k)\in\mathbb{R}^{2k}$ is $C_i+r_i$
for $r_i\in\mathbb{R}^2$.
\end{lemma} 
\begin{proof}
Without loss of generality, assume that the $k$ compact convex figures
are given with centers all lying at the origin.
Let $r=(r_1,\ldots,r_k)\in\mathbb{R}^{2k}$ be a vector such that
the perimeter of their convex hull of $C_i(r)$ is minimized
among all translation vectors in $\mathbb{R}^{2k}$.
Then $-r=(-r_1,\ldots,-r_k)$ is also a vector such that
the convex hull of $C_i(-r)$ has the minimum perimeter.
This is because the two convex hulls are symmetric to the origin.
Since the perimeter function is convex by Lemma~\ref{lem:peri.convex},
the convex hull of $C_i(\mathbf{0})$ also has the minimum perimeter,
where $\mathbf{0}$ is the zero vector (a vector of length zero).
\end{proof}

We are now ready to have a main result.
\begin{theorem} \label{theorem:curve.2.contain}
$\gcc$ is a convex $G_3$-\container of all closed curves of length $2$.
\end{theorem}
\begin{proof}
By Lemma~\ref{lem:G3.curve.length}, any closed curve of length 2
can be contained in $\Gamma^+$ under translation.
Without loss of generality, assume that $\gcc$ is given 
as a part of $\Gamma^+$. 
Let $C$ be a closed curve of length 2 that is contained in $\Gamma^+$
and touches its bottom side, and let $\bar{C}$ be the convex hull of $C$.

Suppose that $C$ crosses the top side of $\gcc$.
Let $s$ be a segment contained in $\bar{C}$ and
connecting the top side and the bottom side of $\gcc$ 
such that the upper endpoint of $s$ lies in the interior of $\bar{C}$.
For each $i=1,2$, let $C_i$ be a rotated and translated copy of $C$ by $2i\pi/3$
such that they are contained in $\Gamma^+$ (by Lemma~\ref{lem:G3.curve.length})
and touch the bottom side of
$\gcc$. If $C_1$ or $C_2$ is contained in $\gcc$,
$\gcc$ is a convex $G_3$-\container of $C$ and
we are done. 

Assume to the contrary that neither $C_1$ nor $C_2$ is contained
in $\gcc$. Then both curves cross the top side of $\gcc$.
For $i=1,2$, let $s_i$ be a line segment contained in the
convex hull of $C_i$ and
connecting the top side and the bottom side of $\gcc$ 
such that the upper endpoint of $s_i$ lies in the interior of 
the convex hull of $C_i$. See Figure~\ref{fig:G3-curve-properties}(a).
Then there is a rotated and translated copy $s'_i$ of $s_i$ by $-2i\pi/3$
such that $s'_i$ is contained in $\bar{C}$.
Let $\chull$ be the convex hull of $s, s'_1$, and $s'_2$.
Since $s, s'_1, s'_2\subset \bar{C}$
and the upper endpoint of $s$ lies in the interior of $\bar{C}$,
$\peri{\chull}<\peri{\bar{C}}\leq \peri{C}=2$.

Now consider a translation of these three segments such that
their midpoints meet at a point, and let $\chull_m$ be the convex hull 
of the three translated segments. 
By Lemma~\ref{lem:k.segments.convexhull}, $\peri{\chull_m}\le \peri{\chull}$.
Let $L_\theta$ denote the slab of minimum width at orientation $\theta$ for 
$0 \leq \theta < \pi$ that contains $\chull_m$. 
Let $d_\theta$ be the width of $L_\theta$.
Consider the three slabs $L_0, L_{\pi/3}$, and $L_{2\pi/3}$ of $\chull_m$. 
Observe that $d_\theta$ for $\theta=0,\pi/3,2\pi/3$ 
is at least height of $\gcc$, which is $2/3$.
Let $H=L_0\cap L_{\pi/3}\cap L_{2\pi/3}$ as shown 
in Figure~\ref{fig:G3-curve-properties}(b).
Then $\chull_m$ is contained in $H$ 
and it touches every side of $H$. 
Since $H$ is a (possibly degenerate) hexagon, 
$\peri{\chull_m}\ge d_0+d_{\pi/3}+d_{2\pi/3} \geq 2$, which can be shown by a proper tiling of copies of $H$ as in Lemma~\ref{lem:tiling}. Thus, $\peri{\chull_m}\ge 2$, contradicting $\peri{\chull_m}\le \peri{\chull}<2$.
\end{proof}

Recall that the smallest-area convex $G_2$-\container $\triangle_1$ and 
$G_4$-\container $\triangle_\beta$ of $\ccset$ are equilateral triangles.
Our $G_3$-\container, $\gcc$, has area smaller than the area of $\triangle_1$, 
but a bit larger than the area of $\triangle_\beta$, 
which sounds reasonable.
 
However, it may look odd that $\gcc$ is not \emph{regular} under any discrete
rotation while $\triangle_1$ and $\triangle_\beta$ are regular 
under rotation by $2\pi/3$.
We show that any convex $G_3$-\container regular under rotation by
$2\pi/3$ or $\pi/2$ has area strictly larger than the area of $\gcc$.

Let $\Lambda$ be a convex $G_3$-\container which is regular under rotation by $2\pi/3$.  
Then $\Lambda$ is a $T$-\container of all unit segments of 
any slope $\theta$ for $0\le\theta<\pi$.
Since $\triangle_1$ is the smallest-area convex $T$-\container of 
the set of all unit segments by Theorem~\ref{thm:pal.kakeya},
the area of $\Lambda$ is at least the area of $\triangle_1$, which is strictly
larger than the area of $\gcc$.

\begin{figure}[t]
	\centering
	\includegraphics[scale=.8]{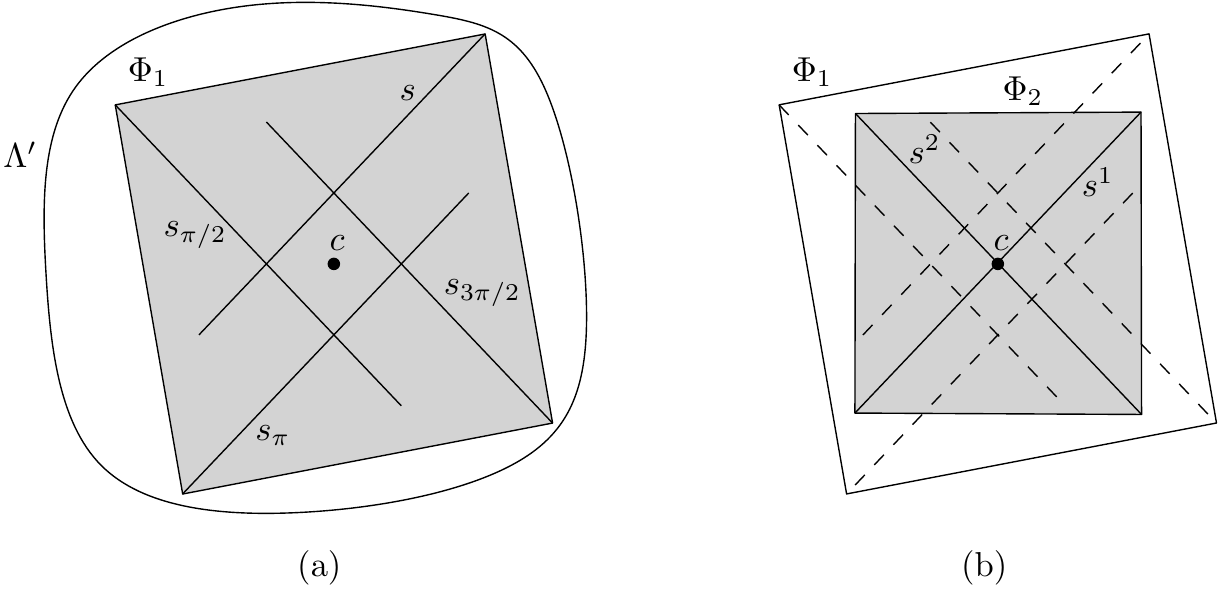}
	\caption{(a) $s_{i\pi/2}$ is contained in $\Lambda'$ for every $i=1, 2, 3$.  (b) The convex hull of $s^1$ and $s^2$ is contained in the convex hull $\Phi_1$ 
	of $s, s_{\pi/2}, s_{\pi},$ and $s_{3\pi/2}$.}
	\label{fig:G3-regular}
\end{figure}

Let $\bar{\Lambda}$ be a convex $G_3$-\container which is regular under rotation by $\pi/2$.
Assume that a unit segment $s$ of slope $\pi/4$ is contained in $\bar{\Lambda}$.
Since $\bar{\Lambda}$ is a $G_3$-\container,
there is a unit segment $s_1$ of one of slopes 
$\{0, \pi/3, 2\pi/3\}$ contained in $\bar{\Lambda}$.
Assume that $s_1$ of slope $\pi/3$ is contained in $\bar{\Lambda}$.
Since $\bar{\Lambda}$ is regular under rotation by $\pi/2$, there are 
unit segments $s'$ of slope $3\pi/4=\pi/4+\pi/2$ and $s'_1$ of slope 
$5\pi/6=\pi/3+\pi/2$ contained in $\bar{\Lambda}$. Thus, the four segments
$s,s',s_1, s'_1$ are contained in $\bar{\Lambda}$.

Let $c$ be the point of symmetry of $\bar{\Lambda}$. 
Let $\chull$ be the convex hull
of the translated copies of $s, s', s_1, s_{1}'$ such that their midpoints are all at $c$.
We will prove that the area $\|\bar{\Lambda}\|$ of $\bar{\Lambda}$ is at least the area $\|\chull\|$ of $\chull$. 
Then $\|\bar{\Lambda}\|\ge\|\chull\|=\sqrt6/4>\|\Gamma_3\|$, 
where $\|\Gamma_3\|$ is the area of $\Gamma_3$.

Suppose that the midpoint of $s$ is not at $c$.
Let $s_{i\pi/2}$ be the copy obtained by rotating $s$ around $c$ by $i\pi/2$ 
for $i=1, 2, 3$.
Since $\bar{\Lambda}$ is regular under the rotation by $\pi/2$, 
$s_{i\pi/2}$ is contained in $\bar{\Lambda}$ for every $i =1,2,3$. 
Since $\bar{\Lambda}$ is convex, the convex hull $\chull_1$ of $s$ and the segments
$s_{i\pi/2}$ for all $i=1, 2, 3$ is contained in $\bar{\Lambda}$,
and thus $\|\bar{\Lambda}\|\ge\|\chull_1\|$. 
See Figure~\ref{fig:G3-regular}(a).

Let $s^1$ be the translated copy 
of $s$ such that the midpoint of $s^1$ is at $c$,
and $s^2$ be the copy of $s^1$ rotated by $\pi/2$ around $c$.
Let $\chull_2$ be the convex hull of $s^1$ and $s^2$.
Since $s^1$ is contained in the convex hull of $s$
and $s_\pi$, and $s^2$ is contained in the convex hull of $s_{\pi/2}$ and
$s_{3\pi/2}$, $\chull_2\subset \chull_1$. See Figure~\ref{fig:G3-regular}(b).
Similarly, the convex hull of the translated copy $\bar{s}$
of $s_1$ with midpoint lying at $c$ and the rotated copy of $\bar{s}$
by $\pi/2$ around $c$ is contained in $\bar{\Lambda}$. 
Thus, we conclude that 
$\|\bar{\Lambda}\|\ge\|\chull\|\ge\sqrt6/4> 0.6>\|\Gamma_3\|$.

We can show this for $s_1$ of slopes $0$ and $2\pi/3$ contained in $\bar{\Lambda}$
in a similar way. 

\section{Covering of triangles under rotation by 120 degrees}
\subsection{Construction}
\label{sec:g3.gt.construction}
Let $\triset$ be the set of all triangles of perimeter $2$. 
We construct a convex $G_3$-\container of $\triset$, 
denoted by $\gt$, from $\gcc$ by shaving off some regions
around the top corners.
Consider an equilateral triangle $\triangle=\triangle XYZ$
of perimeter 2 such that side $YZ$ is vertical, $Y$ lies on the bottom side of 
$\gcc$ and $X$ lies on the left side of $\gcc$. See Figure~\ref{fig:G3-tri-construction}(a-b)
for an illustration. 

Imagine $\triangle$ rotates in a clockwise direction such that $X$ moves along 
the left side and $Y$ moves along the bottom side of $\gcc$.
Let $t$ denote the $x$-coordinate of $X$ and $\theta=\angle{XYO}$. 
Then, $\tan\theta=\sqrt{\frac{2t+1}{-2t-\frac{5}{9}}}$ and 
$Z=\left(\frac{\sqrt{6t+3}+\sqrt{-2t-\frac{5}{9}}+2t}{2},\frac{\sqrt{-6t-\frac{5}{3}}+\sqrt{2t+1}}{2}\right)$ for $t$ varying from $-4/9$ to $-1/3$.
The trajectory of $Z$ forms the top-right boundary of $\gt$
that connects the top side and the right side of $\gcc$.
Thus, the region of $\gcc$ lying above the trajectory is shaved off. 
The top-left boundary of $\gt$ can be obtained similarly. 
Figure~\ref{fig:G3-tri-construction}(c) shows $\gt$.

We show that $\gt$ is convex by showing that 
$\frac{d}{dx}(\frac{dy}{dx})=\frac{d}{dt}(\frac{dy}{dx})/\frac{dx}{dt}\le0$
for $Z = (x(t), y(t))$, and the boundary of $\gt$ has a unique tangent 
at $t=-4/9$ and $-1/3$.
Since the $x$-coordinate of $Z$ increases as $t$ increases, $\frac{dx}{dt}>0$.
Thus, it suffices to show that $\frac{d}{dt}(\frac{dy}{dx})\le0$ for $t$ with $-4/9\le t \le -1/3$.
Observe that
\begin{equation*}
\frac{dy}{dx} = \frac{-3(-6t-\frac{5}{3})^{-\frac{1}{2}}+(2t+1)^{-\frac{1}{2}}}{3(6t+3)^{-\frac{1}{2}}-(-2t-\frac{5}{9})^{-\frac{1}{2}}+2}.
\end{equation*}
We obtain
\begin{equation*}
f(t):= \frac{d}{dt}\left(\frac{dy}{dx}\right)=\frac{(36t+10)f_1(t)-\sqrt{3}(54+108t)f_2(t)-24}{f_1(t)f_2(t)\{2f_1(t)f_2(t)+\sqrt{3}f_1(t)-3f_2(t)\}^2},
\end{equation*} 
where $f_1(t)=\sqrt{-18t-5}$ and $f_2(t)=\sqrt{2t+1}$.
Since $f_1(t)>0$ and $f_2(t)>0$, the denominator of $f(t)$ is positive.
Since the numerator of $f(t)$ is negative, 
$\frac{d}{dt}(\frac{dy}{dx})/\frac{dx}{dt}\le0$.
At $t=-4/9$, $\frac{dy}{dx}=0$, 
which is the slope of the top side of $\gcc$. At $t=-1/3$,
$\frac{dy}{dx}=-\sqrt{3}$, which is 
the slope of the tangent to the right side of $\gcc$ at the same point.
Thus, $\gt$ is convex.

Now we show the area of $\gt$.
The area that is shaved off from $\gcc$ is 
\begin{equation*}
2 \left( \frac{2}{3}\left(\frac{13}{18} -\frac{1}{\sqrt{3}} \right) +\int_{\frac{5}{18}}^{\frac{1}{3}} \sqrt{1-2x}\, dx -\int_{\frac{1}{\sqrt{3}}-\frac{4}{9}}^{\frac{1}{3}} f(x)\, dx \right)
= \frac{1}{81}(86-44\sqrt{3}-3\pi),
\end{equation*}
where $f(x)$ is the function of $\gamma_{DE}$ such that
\begin{equation*}
\int_{\frac{1}{\sqrt{3}}-\frac{4}{9}}^{\frac{1}{3}} f(x)\, dx
= \frac{1}{4} \int_{-\frac{4}{9}}^{-\frac{1}{3}} \left(\sqrt{-6x-\frac{5}{3}}+\sqrt{2x+1}\right)\left(\frac{3}{\sqrt{6x+3}}-\frac{1}{\sqrt{-2x-\frac{5}{9}}}+2\right)\, dx.
\end{equation*}
Thus, $\gt$ has area smaller than $0.5634$. 

\begin{figure}[t]
	\centering
	\includegraphics[width=.9\textwidth]{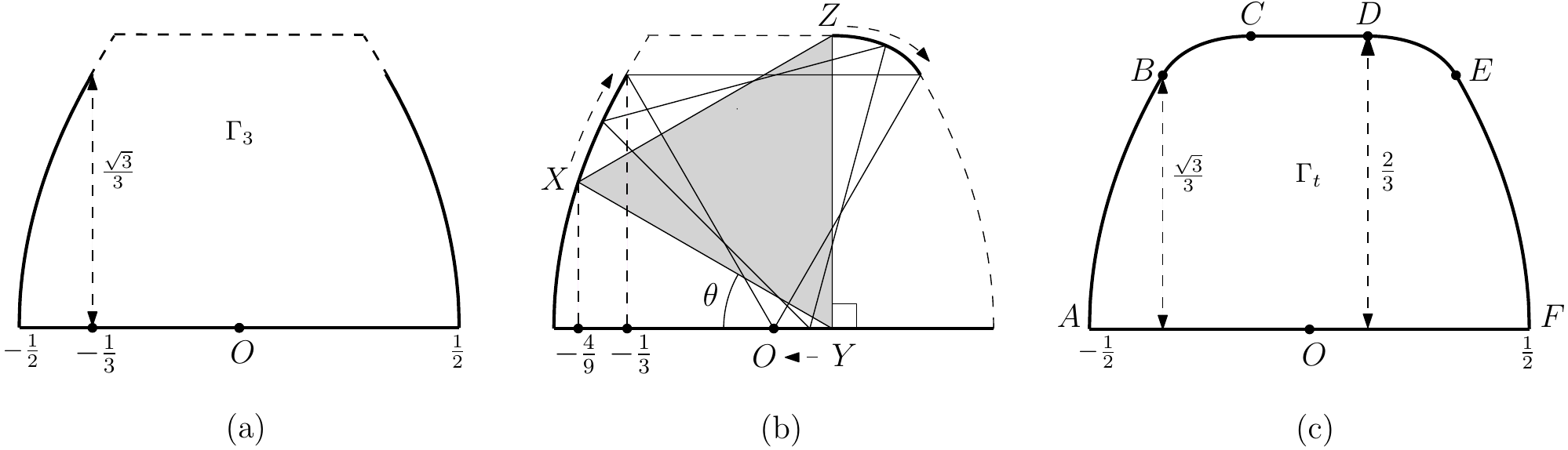}
	\caption{Construction. (a) $\gcc$. (b) The trajectories of $X$, $Y$ and $Z$. (c) A $G_3$-\container $\gt$ of $\triset$.}
	\label{fig:G3-tri-construction}
\end{figure}
\subsection{Covering of triangles of perimeter 2}
We show that $\gt$ is a $G_3$-\container of all triangles of perimeter 2. 
To do this, we first show a few properties
that we use in the course. Let $A,B,C,D,E$, and $F$ be the boundary points of $\gt$
as shown in Figure~\ref{fig:G3-tri-construction}. We denote the boundary curve
of $\gt$ from a point $a$ to a point $b$ in clockwise direction along the boundary 
by $\gamma_{ab}$. 

\begin{figure}[tb]
	\centering
	\includegraphics[width=\textwidth]{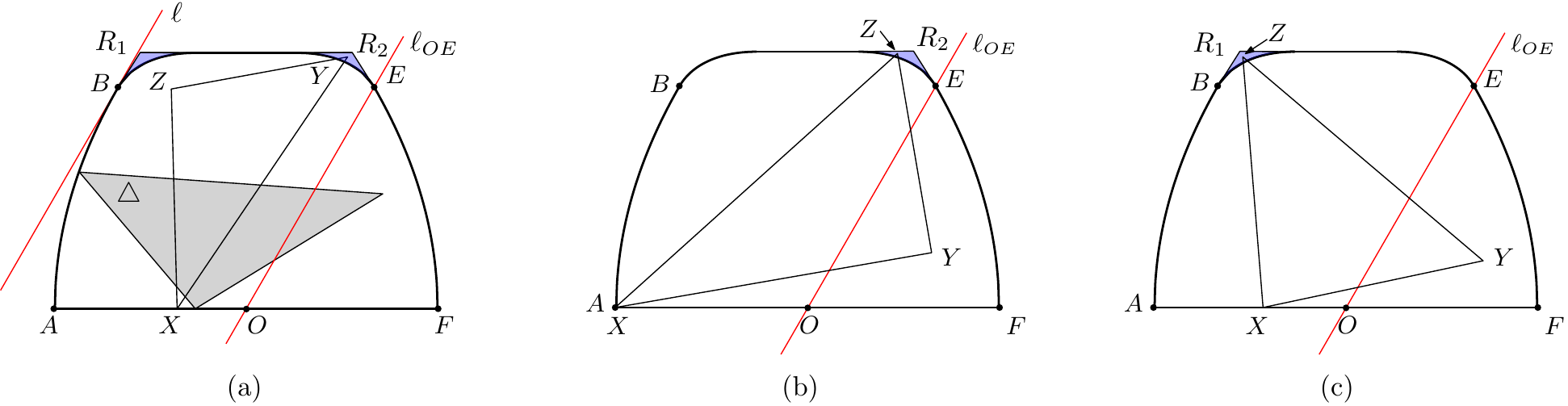}
	\caption{(a) $\triangle{XYZ}$ lying in the left of $\ell_{OE}$. 
	There is a copy $\triangle$	
	that is contained in $\gt$. (b) $\triangle{XYZ}$ contained in $\gcc$
such that $X$ is at $A$, $Y$ is in the right of $\ell_{OE}$,
and $Z\in R_2$. (c) $\triangle{XYZ}$ contained in $\gcc$
such that $X$ is on $\gamma_{OA}$, $Y$ is in the right of $\ell_{OE}$,
and $Z\in R_1$. 
	}
	\label{fig:G3-tri-cases-X}
\end{figure}

\begin{lemma} \label{lem:triangle.left.line}
Let $\triangle{XYZ}$ be a triangle of perimeter 2 contained in $\Gamma^+$. If it is in the left of $\ell_{OE}$
or in the right of $\ell_{OB}$ (including the lines),
$\gt$ is a $G_3$-\container of $\triangle{XYZ}$.
\end{lemma}
\begin{proof}
If $\triangle{XYZ}$ lies in the left of $\ell_{OE}$ (including the line),
$\triangle{XYZ}$ lies in between $\ell_{OE}$ and the line $\ell$ tangent to 
$\Gamma_t$ at $B$.
The two lines are of slope $\pi/3$ and they are
at distance $\sqrt{3}/3$.
Thus there are copies of $\triangle{XYZ}$ rotated by $2\pi/3$ and lying in between 
$\ell_{BE}$ and $\ell_{AF}$.
Among such copies, let $\triangle$ be the one that touches
$\gamma_{FA}$ from above and $\gamma_{AB}$ from right.
Since $\peri{\triangle}=2$, $\triangle$ is contained in $\gt$ by Lemma~\ref{lem:G3.curve.length}.
See Figure~\ref{fig:G3-tri-cases-X}(a). The case of $\triangle{XYZ}$ lying 
in the right of $\ell_{OB}$ can be shown by a copy of the triangle
rotated by $-2\pi/3$.
\end{proof}

In the following, we assume that $\triangle{XYZ}$ is contained in 
$\Gamma_3$ but it is not contained in $\gt$. 
If there is no corner of $\triangle{XYZ}$ lying in the left of $\ell_{OB}$ or 
in the right of $\ell_{OE}$ (including the lines), $\gt$ is a $G_3$-\container of $\triangle{XYZ}$ by Lemma~\ref{lem:triangle.left.line}.
Thus, we assume that $X$ is in the left of $\ell_{OB}$
and $Y$ is in the right of $\ell_{OE}$.
Since $\gt$ is convex,
the remaining corner $Z$ of $\triangle{XYZ}$ lies in 
$\gcc\setminus\gt$.
Let $R_{1}$ and $R_{2}$ denote 
the left and right regions of $\gcc\setminus\gt$, respectively,
as shown in Figure~\ref{fig:G3-tri-cases-X}(a). 
Translate $\triangle{XYZ}$ leftwards horizontally until $X$ or $Z$ hits the left side of $\Gamma_3$. 
If the triangle lies in the left of $\ell_{OE}$ (including the line), we are done by Lemma~\ref{lem:triangle.left.line}.
Thus, we assume that $Z$ is in $R_1\cup R_2$ and $Y$ lies in the right of $\ell_{OE}$.
There are two cases,  either $Z\in R_1$ or $Z\in R_2$. See Figure~\ref{fig:G3-tri-cases-X}(b) and (c).
If $Z$ is in $R_2$, then $X$ is at $A$.

\begin{lemma}\label{lem:triangle.R2}
Let $\triangle{XYZ}$ be a triangle contained in $\gcc$
such that $X$ is at $A$, $Y$ is in the right of $\ell_{OE}$,
and $Z\in R_2$. Then $\gt$ is a convex $G_3$-\container of $\triangle{XYZ}$. 
\end{lemma}
\begin{proof}
Let  $\triangle=\triangle{X'Y'Z'}$ 
be the copy of $\triangle{XYZ}$ rotated by $2\pi/3$ such that $X'$ 
lies at $F$. We show that $\triangle$ is contained in $\gt$.
Assume to the contrary that $\triangle$ is not contained in $\gt$.
Since $\angle ZAF>\pi/6$ and  
$F$ is at distance at least 1 from any point on $\gamma_{AB}$, 
$Z'$ must be contained in $\gt$. See Figure~\ref{fig:G3-tri-lem25}(a) and (b).
\begin{figure}[b]
	\centering
	\includegraphics[width=\textwidth]{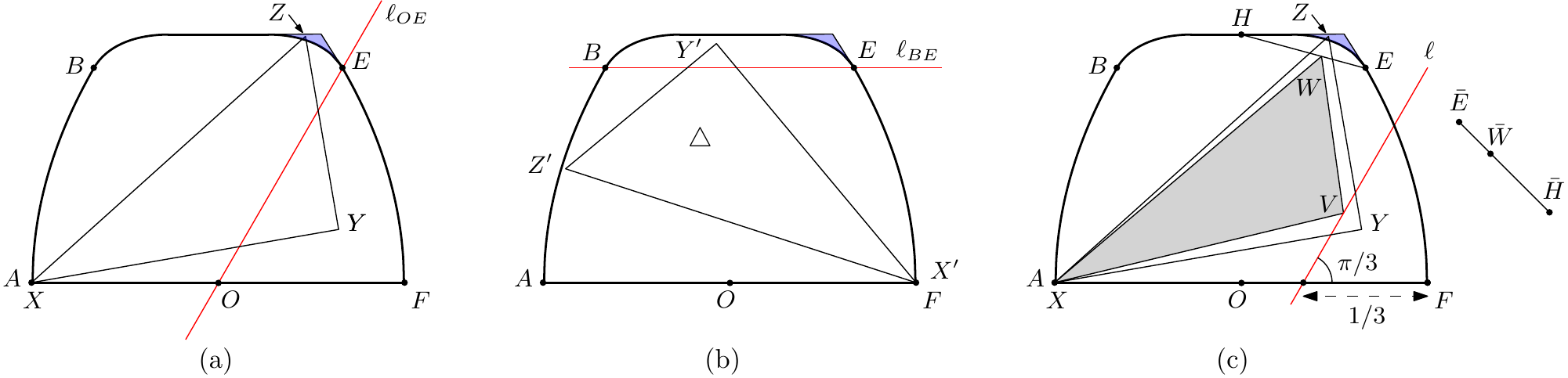}
	\caption{(a) $\triangle{XYZ}$ with $Z\in R_2$. 
	(b) A rotated copy $\triangle{X'Y'Z'}$ of $\triangle{XYZ}$ by $2\pi/3$ with $X'$ lying at $F$. 
	(c) If $Y$ lies in the right of $\ell$, then $\peri{\triangle{XYZ}}>\peri{\triangle{XVW}}\ge 2$.
	}
	\label{fig:G3-tri-lem25}
\end{figure}

If $Y'$ lies on or below $\ell_{BE}$, $\triangle$ is contained in $\gt$ 
by Lemma~\ref{lem:G3.curve.length}.
So assume that $Y'$ lies above $\ell_{BE}$.
Then $Y$ must lie to the right of the line $\ell$ of slope $\pi/3$ 
and passing through the point of $\gamma_{FO}$ at distance $1/3$ from $F$. 
See Figure~\ref{fig:G3-tri-lem25}(c).
Let $H$ be the point at $(0,2/3)$.
Then there is a triangle $\triangle{XVW}$ with $V\in\ell$, and $W\in HE$
such that $\peri{\triangle{XVW}}<\peri{\triangle{XYZ}}$.
Thus, to show a contradiction, 
it suffices to show $\peri{\triangle{XVW}}\ge 2$.
For a point $p$, let $\bar{p}$ denote the point symmetric to $p$ along $\ell$. 
Then $ |VW|=|V\bar{W}|$. 
Since $X$ is at distance at least $7/6$ to any point in $\bar{H}\bar{E}$ 
and at distance at least $5/6$ to any point in $HE$, 
$\peri{\triangle{XVW}}=|XV|+|VW|+|WX|=|XV|+|V\bar{W}|+|WX| \geq 
7/6+5/6=2$.
\end{proof}

We can also show that $\gt$ is a $G_3$-\container of $\triangle{XYZ}$ for 
the remaining case of $Z\in R_1$. 


\begin{lemma}\label{lem:triangle.R1}
Let $\triangle{XYZ}$ be a triangle contained in $\gcc$
such that $X$ is on $\gamma_{OA}$, $Y$ is in the right of $\ell_{OE}$,
and $Z\in R_1$.
Then $\gt$ is a convex $G_3$-\container of $\triangle{XYZ}$.
\end{lemma}

Before proving the lemma, we need a few technical lemmas. 
\begin{lemma} \label{lem:iso.rotate1}
Let $\triangle{XYZ}$ be an isosceles triangle of perimeter $2$ 
such that its base $YZ$ is of length $\geq 2/3$ and 
parallel to the bottom side of $\gt$, and $X$ lies at $O$.
Then $\triangle{XYZ}$ can be rotated in a clockwise direction 
within $\gt$ such that $X$ moves along $\gamma_{OA}$ and
$Y$ moves along $\gamma_{EF}$ until $Y$ meets $F$.
\end{lemma}
\begin{figure}[b]
	\centering
	\includegraphics[width=\textwidth]{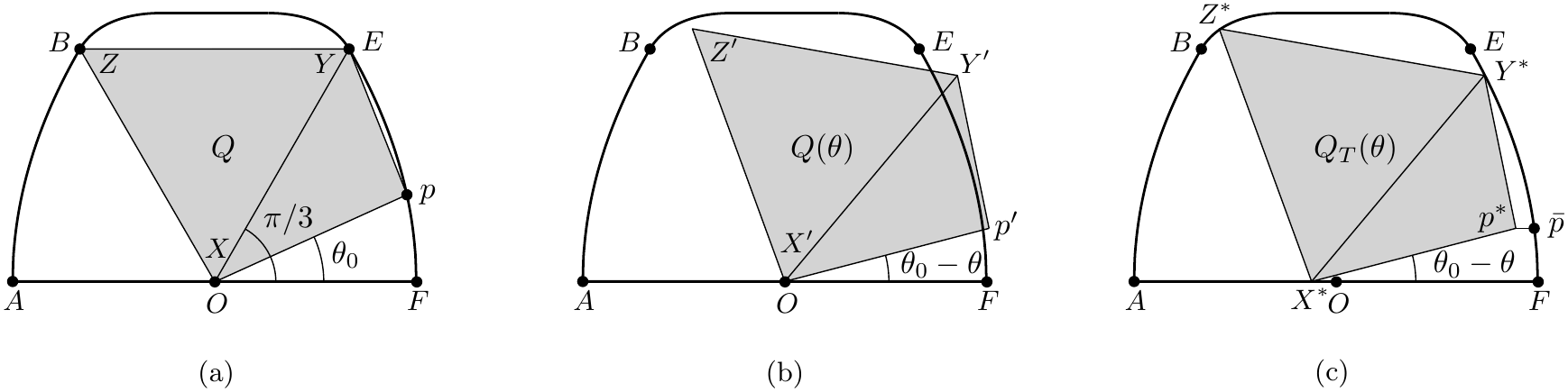}
	\caption{(a) The convex hull $Q$ of an equilateral triangle $\triangle{XYZ}$ and $p$. (b) Rotated copy $Q(\theta)$ of $Q$ by $\theta$ around $O$. (c) Translated copy $Q_T(\theta)$ of $Q(\theta)$ such that $Y^*\in\gamma_{EF}$ and $X^*\in\gamma_{OA}$}
	\label{fig:G3-tri-properties}
\end{figure}
\begin{proof} 
Let $\triangle{XYZ}$ be an equilateral triangle satisfying the conditions
in the lemma statement. Then $|YZ|=2/3$ and $\angle YOF= \pi/3$.
Let $p$ be a point in $\gt$ lying below $\ell_{YO}$, 
and let $Q$ be the convex hull of $p$ and $\triangle{XYZ}$.
Let $\theta_0$ be the angle $\angle pOF$. See Figure~\ref{fig:G3-tri-properties}(a).

We claim that $Q$ can be rotated within $\gt$ in a clockwise direction
such that $Y$ moves along $\gamma_{EF}$ 
and $X$ moves along $\gamma_{OA}$ until $p$ reaches $\ell_{AF}$.
Let $Q(\theta)=X'Y'Z'p'$ be the rotated copy of $Q$ by $\theta$ with $0\le\theta\le\theta_0$ in clockwise
direction around $O$ such that each corner $\kappa'$ of $Q(\theta)$ 
corresponds to the rotated point of $\kappa$ for $\kappa\in\{X,Y,Z,p\}$.
Let $Q_T(\theta)=X^*Y^*Z^*p^*$ be the translated copy of $Q(\theta)$ such that $Y^*\in\gamma_{EF}$, $X^*\in\gamma_{OA}$, and each corner
$\kappa^*$ of $Q_T(\theta)$
corresponds to the translated point of $\kappa'$ for $\kappa'\in\{X',Y',Z',p'\}$.
See Figure~\ref{fig:G3-tri-properties}(b) and (c) for an illustration.

Since $\triangle{XYZ}$ is an equilateral triangle, 
$Z^*$ is contained in $\gt$
as shown in Section~\ref{sec:g3.gt.construction}.
Thus, we show that $p^*\in\gt$.
We may assume that $p \in \gamma_{EF}$.
Let $\bar{p}$ be the intersection between
the horizontal line through $p^*$ and $\gamma_{EF}$.
We show that $f(\theta)=x(\bar{p})-x(p^*)\ge 0$ with $0\le\theta\le\theta_0$, implying $p^*\in\gt$.
Then $f(\theta)=x(\bar{p})-x(p^*)= x(\bar{p})-x(p')+x(p')-x(p^*)=x(\bar{p})-x(p')+x(Y')-x(Y^*)$.
Let $g(\theta_0, \theta)=x(p')-x(\bar{p})$ 
for $0\le\theta\le\theta_0 \le \pi/3$.
Observe that $g(\theta_0, \theta)=x(p')-x(\bar{p})=
h(\theta_0)\cos(\theta_0-\theta)-(\frac{1}{2}-\frac{1}{2}(h(\theta_0)\sin(\theta_0-\theta))^2)$,
where $h(\theta_0)=\frac{1}{1+\cos(\theta_0)}$.
For $\theta_0=\pi/3$, $p'$ is at $Y'$ and $\bar{p}$ is at $Y^*$,
and thus $x(Y')-x(Y^*)=g(\pi/3,\theta)$.
From this, $f(\theta)=g(\pi/3, \theta)-g(\theta_0, \theta)$.
Since $p^*$ is on $\gamma_{EF}$ at $\theta=0$, $f(0)=0$.
Thus, it suffices to show that $g(\theta_0, \theta)$ is 
not decreasing for $\theta_0$ increasing from $\theta$ to $\pi/3$
for any fixed $\theta$ with $0<\theta\le \pi/3$.
\begin{eqnarray*}
\frac{\partial g}{\partial \theta_0} &=& \frac{1}{(1+\cos\theta_0)^3}(g_1(\theta_0)+g_2(\theta_0)) >0,\; \text{where}\\
g_1(\theta_0) &=& \sin\theta_0\sin^2(\theta-\theta_0)+(\sin(\theta-\theta_0)+\sin\theta_1\cos(\theta-\theta_0))(1+\cos\theta_0)\\
&=& \sin\theta_0\sin^2(\theta-\theta_0)+(\cos\theta_0\sin(\theta-\theta_0)+\sin\theta_0\cos(\theta-\theta_0))(1+\cos\theta_0)\\
& & +\sin(\theta-\theta_0)(1-\cos\theta_0)(1+\cos\theta_0) \\
&=& \sin\theta_0\sin^2(\theta-\theta_0)+\sin\theta(1+\cos\theta_0) +\sin(\theta-\theta_0)\sin^2\theta_0\\
&=& (\sin(\theta_0-\theta)-\sin\theta_0)\sin\theta_0\sin(\theta_0-\theta)+\sin\theta+ \sin\theta\cos\theta_0,\;\text{and}\\
g_2(\theta_0) &=& \sin(\theta_0-\theta)(1+\cos\theta_0)(\cos(\theta_0-\theta)-\cos\theta_0).
\end{eqnarray*}
Observe that $\sin\theta\cos\theta_0 \ge 0$. Since $0<\sin\theta_0\sin(\theta_0-\theta)<1$, 
it suffices to show that
$(\sin(\theta_0-\theta)-\sin\theta_0)+\sin\theta > 0$ to show $g_1(\theta_0) > 0$.
Let $\bar{g}(\theta_0)=\sin(\theta_0-\theta)-\sin\theta_0+\sin\theta$.
Since $\frac{\partial \bar{g}}{\partial \theta_0} = \cos(\theta_0-\theta)-\cos\theta_0>0$
and $\bar{g}(\theta) = 0$,
we have $\bar{g}(\theta_0) > 0$, implying $g_1(\theta_0) > 0$.
Observe that $1+\cos\theta_0>0$ and $(\cos(\theta_0-\theta)-\cos\theta_0) > 0$,
so we have $g_2(\theta_0) > 0$.
Thus, $\frac{\partial g}{\partial \theta_0} > 0$.
Since $p^*$ is contained in $\gt$ at $\theta=\theta_0$, $p^*$ is contained in $\gt$.

\begin{figure}[t]
	\centering
	\includegraphics[width=\textwidth]{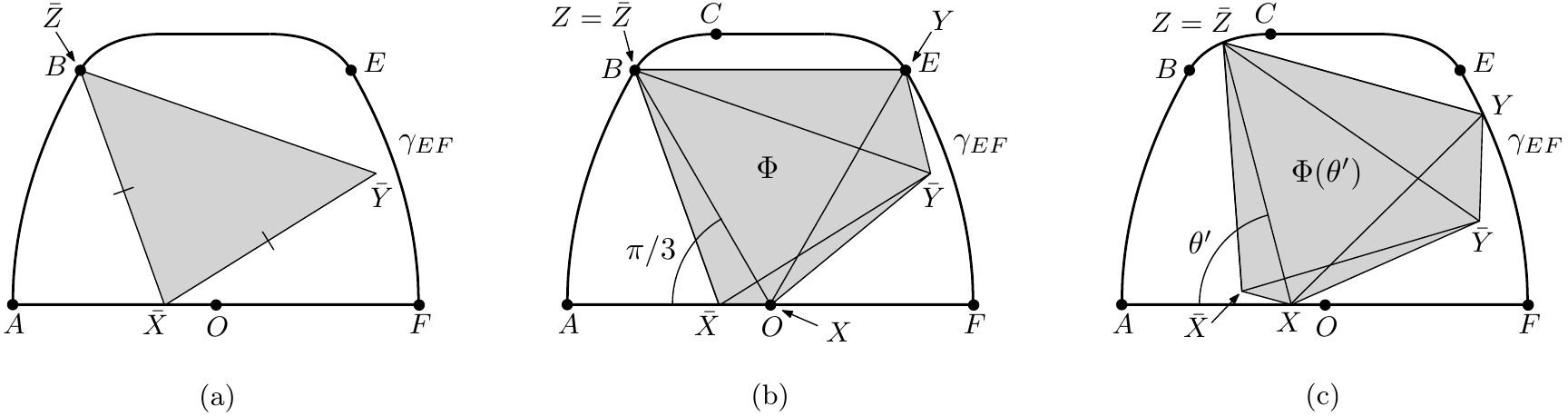}
	\caption{(a) An isosceles triangle $\triangle \bar{X}\bar{Y}\bar{Z}$. 
	(b) Convex hull $\chull$ of $\triangle \bar{X}\bar{Y}\bar{Z}$ and $\triangle XYZ$.
	(c) The rotated copy $\chull(\theta')$ of $\chull$ by $\theta'-\pi/3$.}
	\label{fig:G3-tri-properties_second}
\end{figure}
By Lemma~\ref{lem:G3.curve.length}, any triangle of perimeter $2$
lying in between $\ell_{BE}$ and $\ell_{AF}$ can be translated 
such that it is contained in $\gt$.
Thus, we consider isosceles triangles with one corner lying on 
$\gamma_{FA}$ and one corner lying above $\ell_{BE}$.
Let $\triangle{\bar{X}\bar{Y}\bar{Z}}$ be the isosceles triangle of perimeter $2$ contained in $\gt$
such that its base $\bar{Y}\bar{Z}$ is of length $\ge 2/3$, $\bar{Z}$ is at $B$, and $\bar{X}\in\gamma_{OA}$. 
Observe that $|X\bar{X}| \le 1/3$;
if $|X\bar{X}|>1/3$, $\bar{Y}$ is not contained in $\gt$ because 
$\triangle{\bar{X}\bar{Y}\bar{Z}}$ is an isosceles triangle of perimeter $2$ 
with base length $|\bar{Y}\bar{Z}|\ge 2/3$.
See Figure~\ref{fig:G3-tri-properties_second}(a).

Let $\chull$ be the convex hull of $\triangle{XYZ}$ and $\triangle{\bar{X}\bar{Y}\bar{Z}}$. See Figure~\ref{fig:G3-tri-properties_second}(b).
Imagine that $\chull$ rotates in a clockwise
direction such that $X$ moves along $\gamma_{OA}$ and $Y$ moves 
along $\gamma_{EF}$ until $\bar{Y}$ reaches $\gamma_{FA}$.
Then $Z$ moves along $\gamma_{BC}$, and then moves 
into the interior of $\gt$ during the rotation because $\triangle{XYZ}$ is an equilateral
triangle with $|XZ|=2/3$.
Let $\chull(\theta')$ denote the rotated copy of $\chull$ by $\theta'-\pi/3$,
where $\theta'=\angle{ZXA}$.
Then $\theta'$ increases monotonically from $\pi/3$ during the rotation. 
See Figure~\ref{fig:G3-tri-properties_second}(c).

Consider the rotation for $\theta'$ from $\pi/3$ to $\pi/2$.
Observe that $\bar{Y}$ is contained in $\gt$ during the rotation 
by the argument in the first paragraph of the proof.
Since $Z=\bar{Z}$ and $Z\in\gt$ during the rotation (shown in Section~\ref{sec:g3.gt.construction}),
$\bar{Z}$ is also in $\gt$.

Now we claim that $\bar{X}$ is contained in $\gt$ during the rotation.
For any $\theta'$ with $\pi/3<\theta' \le \pi/2$, $\bar{X}$ lies above 
$\ell_{OA}$, but it lies below $\ell_{BE}$ because $\angle \bar{X}FA\le\pi/6$
for $\theta'\le \pi/2$.
Since $|\bar{X}F|\leq |\bar{X}X|+|XF| \leq |\bar{X}X|+|XY| \le 1$ 
as $|\bar{X}X|\le 1/3$, $\bar{X}$ is contained in $\gt$.
Thus, $\triangle{\bar{X}\bar{Y}\bar{Z}}$ is contained in $\gt$ during the rotation.
Thus the lemma holds.
\end{proof}
\begin{figure}[b]
	\centering
	\includegraphics[width=.6\textwidth]{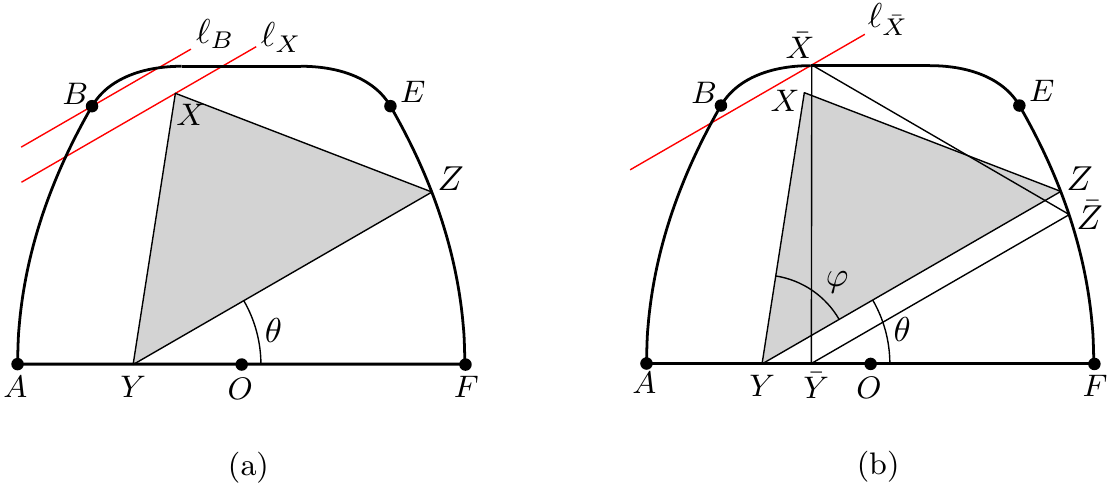}
	\caption{(a) $\ell_X$ lies below $\ell_B$ or $\ell_X=\ell_B$. (b) $X$ lies below or on $\ell_{\bar{X}}$.}
	\label{fig:lem2425}
\end{figure}
\begin{lemma} 
\label{lem:LX.BF.intersect}
Let $\triangle{XYZ}$ be an isosceles triangle of perimeter $2$ 
such that its base $YZ$ is of length $\ge 2/3$, $Y\in\gamma_{OA}$, 
and $Z\in\gamma_{EF}$. The line through $X$ and 
parallel to $YZ$ intersects $\gamma_{OB}$.
\end{lemma}
\begin{proof}
Let $\ell_X$ and $\ell_{B}$ be the lines parallel to $YZ$ such that
$\ell_X$ passes through $X$ and $\ell_{B}$ passes through $B$.
See Figure~\ref{fig:lem2425}(a).
Let $\theta=\angle{ZYF}$ and let $y$ denote the $y$-coordinate of $Z$.
Then $0 \leq y \leq \sqrt{3}/3$ and 
$\sin^{-1}y \leq \theta \leq \sin^{-1}(\frac{3y}{2})$.
Let $f(y, \theta)$ be the distance between $\ell_{B}$ and $X$,
which is $f(y, \theta) =(\frac{\sqrt{3}}{3}-y)\cos\theta+(\frac{5-3y^2}{6})\sin\theta-\sqrt{1-\frac{y}{\sin\theta}}$.

We will show that 
for $y$ with $0 \leq y \leq \sqrt{3}/3$ 
and for $\theta$ with $\sin^{-1}y \leq \theta \leq \sin^{-1}(\frac{3y}{2})$ (or equivalently $\frac{2\sin\theta}{3} \le y \le \min\{\sqrt{3}/3,\sin\theta\}$ and $0\le\theta\le\pi/3$),
\begin{equation*}
\frac{\partial f}{\partial y} = -y\sin\theta-\cos\theta+\frac{1}{2\sqrt{\sin^2\theta-y\sin\theta}}\geq 0.
\end{equation*}

We have
$\frac{\partial^2 f}{\partial y^2} = \sin\theta\left(\frac{1}{4\sin\theta(\sin\theta-y)^{3/2}}-1\right).$
Since $0\le\sin\theta-y \le \frac{\sin\theta}{3}$, we have $0\le4\sin\theta(\sin\theta-y) \le \frac{4\sin^2\theta}{3} \le 1$.
Thus, $0\le4\sin\theta(\sin\theta-y)^{3/2} \le 1$, implying $\frac{\partial^2 f}{\partial y^2} \ge 0$.
So it suffices to show $\frac{\partial f}{\partial y} \ge 0$ for $y=\frac{2\sin\theta}{3}$.
\begin{equation*}
\frac{\partial f}{\partial y}\mid_{y=\frac{2\sin\theta}{3}}=-\frac{2\sin^2\theta}{3}-\cos\theta+\frac{\sqrt{3}}{2\sin\theta} 
= \frac{1}{\sin\theta}\left(-\frac{2\sin^3\theta}{3}-\frac{\sin2\theta}{2}+\frac{\sqrt3}{2}\right)
\end{equation*}

Let $g(\theta)=-\frac{2\sin^3\theta}{3}-\frac{\sin2\theta}{2}+\frac{\sqrt3}{2}$. Then
\begin{eqnarray}
\frac{\partial g}{\partial \theta}&=-\cos2\theta-2\sin^2\theta\cos\theta\\
&=-\cos2\theta+2\cos^3\theta-2.
\end{eqnarray}

For $\theta \le\pi/4$, we have $\frac{\partial g}{\partial \theta} \le 0$ by equation (1).
For $\theta$ with $\pi/4 < \theta \le\pi/3$, we have $2\cos^3\theta \le 1$ and $-\cos2\theta \le 1$ so $\frac{\partial g}{\partial \theta} \le 0$ by equation (2).
Therefore, $g(\theta) \ge g(\pi/3)=0$, implying $\frac{\partial f}{\partial y} \ge 0$.

Now we show that $f(\frac{2\sin\theta}{3}, \theta) \ge 0$ for $\theta$ with $0 \le \theta \le \pi/3$.
Let $h(\theta):= f(\frac{2\sin\theta}{3}, \theta)$. Then
\begin{equation*}
h(\theta)=\frac{5\sin\theta}{6}-\frac{2\sin^3\theta}{9}+\frac{\cos\theta}{\sqrt3}-\frac{2\sin\theta\cos\theta}{3}-\frac{1}{\sqrt3}.
\end{equation*}
Let $t=\sin\theta$. Then $\sqrt{1-t^2}=\cos\theta$ and $0\le t \le \frac{\sqrt3}{2}$.
\begin{eqnarray*}
h(\theta)=0 
&\Rightarrow& -\frac{2}{9}t^3+\frac{5}{6}t-\frac{1}{\sqrt3}=\left(\frac{2}{3}t-\frac{1}{\sqrt3}\right)\sqrt{1-t^2} \\ 
&\Rightarrow & \left(\frac{2}{3}t-\frac{1}{\sqrt3}\right)\left(-\frac{1}{3}t^2-\frac{1}{2\sqrt3}t+1\right)=\left(\frac{2}{3}t-\frac{1}{\sqrt3}\right)\sqrt{1-t^2} \\
&\Rightarrow & \left(-\frac{1}{3}t^2-\frac{1}{2\sqrt3}t+1\right)^2=1-t^2 \;\text{or}\; t=\frac{\sqrt3}{2} \\
&\Rightarrow & \frac{1}{9}t^4+\frac{1}{3\sqrt3}t^3+\frac{5}{12}t^2-\frac{1}{\sqrt3}t=0\; \text{or}\; t=\frac{\sqrt3}{2} \\
&\Rightarrow & t(t-\frac{\sqrt3}{2})(\frac{1}{9}t^2+\frac{1}{2\sqrt3}t+\frac{2}{3})\; \text{or}\; t=\frac{\sqrt3}{2} \\
&\Rightarrow & t= 0\; \text{or}\; t=\frac{\sqrt3}{2} \; \Leftrightarrow \; \theta= 0\; \text{or}\; \theta=\frac{\pi}{3}
\end{eqnarray*}

Since $h(\pi/6)=8/9-\sqrt3/2>0$ and $h$ is continuous with roots only at $0$ and $\pi/3$, we have $h(\theta) \ge 0$ for $\theta$ with $0 \le \theta \le \pi/3$.

Observe that $f(y, \theta)$ is minimized at $y = \frac{2\sin\theta}{3}$ for a fixed $\theta$.
In other words, $f(y, \theta)$ is minimized when $\triangle{XYZ}$ is an equilateral triangle.
Since $f(y, \theta) \ge 0$ for $y = \frac{2\sin\theta}{3}$,
any isosceles triangle $\triangle XYZ$ such that 
$\ell_X$ does not intersect $\gamma_{OB}$ has 
height larger than $\sqrt{3}/3$.
Thus, $\ell_X$ intersects $\gamma_{OB}$.
\end{proof}

\begin{lemma} \label{lem:iso.rotate2} 
Let $\triangle{XYZ}$ be an isosceles triangle of perimeter $2$ 
such that its base $YZ$ is of length $\geq 2/3$, 
$X$ lies above or on $\ell_{YZ}$, $Y\in\gamma_{OA}$, 
and $Z\in\gamma_{EF}$. Then $\triangle{XYZ}$ is contained in $\gt$. 
\end{lemma}
\begin{proof} 
Let $\triangle{XYZ}$ an isosceles triangle satisfying the conditions in the lemma
statement with $\varphi=\angle{XYZ}$ and $\theta=\angle{ZYF}$.
Observe that $0 \le \varphi, \theta \le \pi/3$.
Let $f(\varphi)$ be the $y$-coordinate of $X$.
Then $f(\varphi)=\frac{\sin(\varphi+\theta)}{1+\cos\varphi}$.
Since $f'(\varphi) \ge 0$, $f(\varphi)$ is maximized at $\varphi = \pi/3$, 
implying that $f(\varphi)$ is maximized when $\triangle{XYZ}$ is the equilateral triangle.

Let $\triangle{\bar{X}\bar{Y}\bar{Z}}$ be the equilateral triangle such that 
$\bar{Y}\in\gamma_{OA}$, $\bar{Z}\in\gamma_{EF}$, and $\bar{Y}\bar{Z}$ 
is parallel to $YZ$.
Let $\ell_{\bar{X}}$ be the line parallel to $\bar{Y}\bar{Z}$ and passing through $\bar{X}$.
See Figure~\ref{fig:lem2425}(b).
Then $X$ lies below or on $\ell_{\bar{X}}$ by the proof in Lemma~\ref{lem:LX.BF.intersect}, 
and its $y$-coordinate $f(\varphi)$ is smaller than or equal to the 
$y$-coordinate of $\bar{X}$ by the argument in the previous paragraph.
Since $\bar{X}$ is on the boundary of $\gt$, $X$ is contained in $\gt$.
\end{proof}
\begin{figure}[t]
	\centering
	\includegraphics[width=\textwidth]{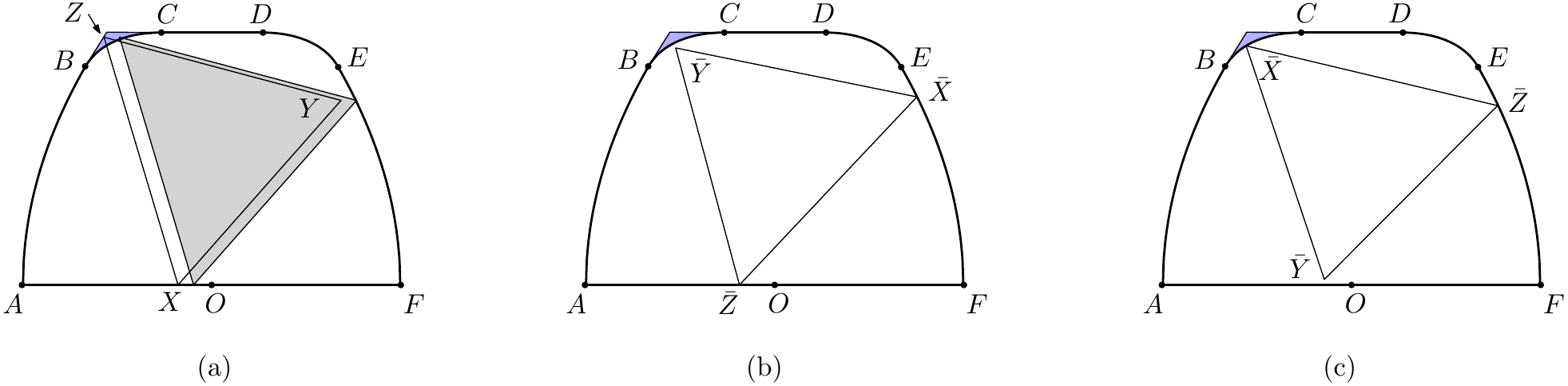}
	\caption{(a) $Z\in R_1$ and $ZX$ is the longest side. 
	(b) A copy $\triangle{\bar{X}\bar{Y}\bar{Z}}$ of $\triangle{XYZ}$ rotated by $2\pi/3$. (c) A copy $\triangle{\bar{X}\bar{Y}\bar{Z}}$ of $\triangle{XYZ}$ rotated by $-2\pi/3$.}
	\label{fig:G3-tri-cases-Z}
\end{figure}

Now we are ready to prove Lemma~\ref{lem:triangle.R1}.
\begin{proof}
Translate $\triangle{XYZ}$
to the right until $Y$ meets $\gamma_{EF}$.
See Figure~\ref{fig:G3-tri-cases-Z}(a). 
If $Z\in\gt$, then $\triangle{XYZ}\subset\gt$ and we are done.

Suppose that $Z\in R_1$.
If $X\in\gamma_{FO}$, $\triangle{XYZ}$ lies in the right of $\ell_{OB}$ (including the line), and by Lemma~\ref{lem:triangle.left.line}, $\gt$ is a $G_3$-\container of
$\triangle{XYZ}$. 
So assume that $X\in\gamma_{OA}$, and thus $|XY| \geq 1/2$. 
There are three cases
that the longest side of $\triangle{XYZ}$ is (1) $XZ$, (2) $YZ$, or (3) $XY$.
For each case, we show that there is a rotated copy of $\triangle{XYZ}$
that is contained in $\gt$.

Consider case (1) that $XZ$ is the longest side. There are two subcases, 
$|XY| \geq |YZ|$ or $|XY| < |YZ|$.
Suppose $|XY| \geq |YZ|$.
Let $\triangle{\bar{X}\bar{Y}\bar{Z}}$ be the copy of $\triangle{XYZ}$ rotated by $2\pi/3$ 
such that $\bar{Z}\in\gamma_{FA}$ and $\bar{X}$ lies on the right side of $\Gamma$. 
Since $\angle{ZXF} > \pi/3$, $\bar{Z}$ is the lowest corner.
So $\bar{X}$ lies on the right side of $\Gamma^+$.
Since $\angle{\bar{Y}\bar{Z}F} \le 2\pi/3$, 
$\bar{Y}$ lies in the right of $\ell_{OB}$ (including the line).
If $\bar{Z}\in\gamma_{FO}$,
$\triangle{\bar{X}\bar{Y}\bar{Z}}$ lies in the right of $\ell_{OB}$ (including the line),
and by Lemma~\ref{lem:triangle.left.line}, $\gt$ is a $G_3$-\container of $\triangle{XYZ}$.
Thus, assume that $\bar{Z}\in\gamma_{OA}$.
If $\bar{X}$ lies in the left of $\ell_{OE}$ (including the line), then 
$\triangle{\bar{X}\bar{Y}\bar{Z}}$ lies in the left of $\ell_{OE}$ (including the line), 
and by Lemma~\ref{lem:triangle.left.line}, $\gt$ is a $G_3$-\container of
$\triangle{XYZ}$.
So assume $\bar{X}\in\gamma_{EF}$.
See Figure~\ref{fig:G3-tri-cases-Z}(b).
Let $Y^{*}$ be a point such that $\triangle{\bar{X}Y^{*}\bar{Z}}$ is a isosceles triangle of perimeter 2 
with base $\bar{Z}\bar{X}$ and lies in the left of $\ell_{\bar{Z}\bar{X}}$ 
(including the line).
Then both $\bar{Y}$ and $Y^{*}$ are on the same ellipse curve with foci $\bar{X}$ and $\bar{Z}$.
By Lemma~\ref{lem:iso.rotate2}, $Y^{*}\in\gt$.
Since $|\bar{X}\bar{Y}| \geq |\bar{Y}\bar{Z}|$, $\bar{Y}$ lies on or below the axis line of the ellipse through $Y^*$.
Since $\bar{Y}$ lies on or below the tangent line of the ellipse at $Y^*$, by Lemma~\ref{lem:LX.BF.intersect}, $\bar{Y}\in\gt$ and thus $\triangle{\bar{X}\bar{Y}\bar{Z}}\subset\gt$.

Suppose now $|XY| < |YZ|$.
Let $\triangle{\bar{X}\bar{Y}\bar{Z}}$ be the copy of $\triangle{XYZ}$ rotated by $-2\pi/3$ 
such that $\bar{X}\in\gamma_{BD}$ and $\bar{Z}\in\gamma_{EF}$.
Since $\angle{ZXF} < 2\pi/3$, $\bar{X}$ is the highest corner.
Let $Y^{*}$ be a point such that $\triangle{\bar{X}Y^{*}\bar{Z}}$ is a isosceles triangle with base $\bar{Z}\bar{X}$ and lies below $\ell_{\bar{Z}\bar{X}}$.
Then both $\bar{Y}$ and $Y^{*}$ are on the same ellipse curve with foci $\bar{X}$ 
and $\bar{Z}$.
By Lemma~\ref{lem:iso.rotate1}, $Y^{*}\in\gt$.
Since $|\bar{X}\bar{Y}| < |\bar{Y}\bar{Z}|$, $\bar{Y}$ lies on or above the axis line of 
the ellipse through $Y^*$.
Since $\bar{Y}$ lies on or above the line tangent to the ellipse at $Y^{*}$, 
$\bar{Y}$  lies above $\ell_{FA}$.
Since $XZ$ is the longest side, $\bar{Y}\in\gt$ and thus $\triangle{\bar{X}\bar{Y}\bar{Z}}\subset\gt$.

Now consider case (2) that $YZ$ is the longest side. Translate $\triangle{XYZ}$ 
such that $Y\in\gamma_{EF}$ and $Z\in\gamma_{BC}$.
If $|XY| \ge |ZX|$, then $X\in\gt$ by the argument for case (1).
So assume that $|XY|< |ZX|$.
Let $\triangle{\bar{X}\bar{Y}\bar{Z}}$ be the copy of $\triangle{XYZ}$ rotated by $-2\pi/3$ 
such that $\bar{Y}\in\ell_{FA}$ and $\bar{Z}$ lies on the right side of 
$\Gamma^+$.
Then $\angle{\bar{Z}\bar{Y}F} \leq \pi/3$. 
If $\triangle{\bar{X}\bar{Y}\bar{Z}}$ lies to the right of  $\ell_{OB}$ or to the left of $\ell_{OE}$ (including the lines), 
$\gt$ is a $G_3$-\container of the triangle by Lemma~\ref{lem:triangle.left.line}.
So assume $\bar{Z}\in\gamma_{EF}$ and $\bar{Y}\in\gamma_{OA}$.
Since $|XY|<|ZX|$, $\bar{X}\in\gt$ by the argument for case (1),
and thus $\triangle{\bar{X}\bar{Y}\bar{Z}}\subset\gt$.

Finally, consider case (3) that $XY$ is the longest side. Translate $\triangle{XYZ}$ 
such that $X\in\gamma_{OA}$ and $Y\in\gamma_{EF}$.
If $|XZ| \leq |YZ|$, then $Z\in\gt$ by the argument for case (1).
So assume that $|XZ|>|YZ|$.
Let $\triangle{\bar{X}\bar{Y}\bar{Z}}$ be the copy of $\triangle{XYZ}$ rotated by $2\pi/3$ 
such that $\bar{X}$ lies on the right side of $\Gamma^+$ and $\bar{Z}\in\gamma_{FA}$.
Since $\angle{ZXF} > \pi/3$, $\bar{Z}$ is the lowest corner and $\angle{\bar{X}\bar{Z}F} \leq \pi/3$. 
If $\triangle{\bar{X}\bar{Y}\bar{Z}}$ lies to the right of $\ell_{OB}$ or to the left of $\ell_{OE}$
(including the lines), 
$\gt$ is a $G_3$-\container of the triangle by Lemma~\ref{lem:triangle.left.line}.
So assume $\bar{X}\in\gamma_{EF}$ and $\bar{Z}\in\gamma_{OA}$.
Since $\bar{Y}\bar{Z}$ is the shortest side, $\bar{Y}$ lies below $\ell_{CD}$. Thus, it belongs to
case (2).
\end{proof}

Combining Lemmas~\ref{lem:triangle.left.line},~\ref{lem:triangle.R2}, and~\ref{lem:triangle.R1},
we have the following result. 
\begin{theorem}  \label{thm:triangle.G3}
$\gt$ is a convex $G_3$-\container of all triangles of perimeter $2$.
\end{theorem}

\section{Conclusion}
We considered the smallest-area covering of planar objects of perimeter $2$ 
allowing translation and discrete rotations of multiples of $\pi, \pi/2,$ and $2\pi/3$.
We gave a geometric and elementary proof of the smallest-area convex \containers
for translation and rotation of $\pi$ while showing convex coverings 
for the other discrete rotations. 
We also gave the smallest-area convex \containers of all unit segments 
under translation and discrete rotations $\pi/k$ for all integers $k\ge 3$.
Open problems include the optimality proof of the equilateral triangle
covering for the rotation of multiples of $\pi/2$, 
and the smallest-area coverings allowing other discrete rotations 
with clean mathematical solutions.

\newpage
\section*{Appendix}
We give a proof of Corollary~\ref{cor:t.container.worm} using 
a reflecting argument similar to the proof of Theorem~\ref{thm:G2.main}.
\begin{proof}
Given a worm $\gamma$, let $\gamma'$ be its largest scaled copy contained in $\triangle_1$.
It suffices to show that the length of $\gamma'$ is not shorter than $1$.

Observe that $\gamma'$ touches all three edges of $\triangle_1$, 
possibly touching two edges simultaneously at a vertex, since 
otherwise we can translate $\gamma'$ into the interior of $\triangle_1$
and get a larger scaled copy of $\gamma'$ contained in $\triangle_1$.
If $\gamma'$ touches a vertex of $\triangle_1$, it connects the vertex 
and the edge opposite to the vertex. Thus, its length is not shorter than 
the height $1$ of $\triangle_1$.

Consider the case that $\gamma'$ touches the three edges of $\triangle_1$ at 
three distinct points $a$, $b$, $c$ along $\gamma'$ in the order. 
Then the length of $\gamma'$ is not shorter than the two-leg polyline 
$abc$ with the joint point $b$.  Now, $abc$ has the same length as $abc'$, 
where $c'$ is the reflection of $c$ with respect to the edge $e$ of 
$\triangle_1$ containing $b$.  Consider the union of $\triangle_1$ and the reflected copy of $\triangle_1$ by $e$. The union is a rhombus of 
height $1$ containing $a$ and $c'$ on opposite parallel edges.  
Thus the length of $abc'$ is at least $1$, and we complete the proof.
\end{proof}
\end{document}